\documentclass[12 pt]{amsart}
\usepackage[english]{babel}
\usepackage{indentfirst}
\usepackage{color}
\usepackage{graphicx}
\usepackage{newlfont}
\usepackage{soul}
\usepackage{amssymb}

\usepackage{amsmath}

\usepackage{latexsym}
\usepackage{verbatim}
\usepackage{hyperref}
\usepackage{amsthm}
\usepackage{mathrsfs}
\usepackage{pifont}
\usepackage{textcomp}

\usepackage{pictexwd,dcpic}
\usepackage{epsfig}
\usepackage{tikz}
\usetikzlibrary{backgrounds}
\usepackage{arydshln}

\theoremstyle{plain}

\newcommand{\IC}{\mathbb{C}}

\newcommand{\eeq}{\mathrel{\mathop:}=}

\newtheorem{theorem}{Theorem}[section]

\newtheorem*{theorem*}{Theorem}
\newtheorem{proposition}[theorem]{Proposition}
\newtheorem*{proposition*}{Proposition}
\newtheorem*{lemma*}{Lemma}

\newtheorem{corollary}[theorem]{Corollary}
\newtheorem{lemma}[theorem]{Lemma}
\newtheorem{question}{Question}

\theoremstyle{definition}

\theoremstyle{remark}
\newtheorem{remark}[theorem]{Remark}

\title[Vortices in the plane]{On Some Configurations of Oppositely Charged Trapped Vortices in the Plane}

\author[E Dufresne]{Emilie Dufresne}
\address{Department of Mathematics, University of York, York, YO10 5DD, UK}
\email{emilie.dufresne@york.ac.uk}

\author[HA Harrington]{Heather A Harrington}
\address{Mathematical Institute, University of Oxford, Andrew Wiles Building, Radcliffe Observatory Quarter, Woodstock Road, Oxford, OX2 6GG}
\email{harrington@maths.ox.ac.uk}

\author[PG Kevrekidis]{Panayotis G Kevrekidis}
\address{Department of Mathematics and Statistics, University of Massachusetts Amherst, Amherst, MA 01003-4515, USA}
\email{kevrekid@math.umass.edu
}

\author[P Tripoli]{Paolo Tripoli}
\address{School of Mathematical Sciences, University of Nottingham, University Park, Nottingham, NG7 2RD}
\email{paotripoli@gmail.com}

\date{\today}



\begin{document}

\begin{abstract}
  Our aim in the present work is to identify all the possible standing
  wave configurations involving few vortices of different charges
  in an atomic Bose-Einstein condensate (BEC). In this effort, we deploy
  the use of a computational algebra approach in order to identify
  stationary multi-vortex states with up to 6 vortices. The
  use of invariants  and symmetries enables deducing a set of
  equations in
  elementary symmetric polynomials, which can then be fully solved via
  computational algebra packages within Maple. We retrieve a number
  of previously identified configurations, including collinear ones
  and polygonal (e.g. quadrupolar and hexagonal) ones. However,
  importantly, we also retrieve a configuration with $4$ positive
  charges and $2$ negative ones which is unprecedented, to the best
  of our knowledge, in BEC studies. We corroborate these predictions
  via numerical computations in the fully two-dimensional PDE system
  of the Gross-Pitaevskii type which characterizes the BEC at the
  mean-field level.
\end{abstract}

\maketitle

%%%%%%%%%%%%%%%%%%%%%%%%%%%%%%%%%%%%%%%%%%%%%%%
%
%       SECTION: Introduction
%
%%%%%%%%%%%%%%%%%%%%%%%%%%%%%%%%%%%%%%%%%%%%%%%

	\section{Introduction}
	
        The explosion of interest in the theme of atomic Bose-Einstein  condensates (BECs) \cite{becbook1,becbook2,siambook} has had  significant implications in the study of associated nonlinear
        coherent structures, including vortices, as well as
        vortex lines and vortex rings~\cite{pismen}. In particular,
        settings involving the emergence and precessional dynamics
        of one or few vortices (see e.g.~\cite{And2000,Hal2001.PRL86.2922,Bre2003.PRL90.100403,NeelyEtAl10,dshall,dshall1,dshall3} for
        some typical examples), as well as the exploration of
        higher charged vortices and their potential decay to
        lower charged ones (see, e.g.,~\cite{S2Ket2,Iso2007.PRL99.200403})
        have been topics that garnered considerable interest within
        the atomic and nonlinear communities and motivated
        numerous associated experiments. Vortical patterns not
        only in two but also in higher dimensions (e.g., filaments in
        the form of lines and rings in 3d) were also produced by
        means of dynamical instabilities such as the extensively
        studied transverse
        instability~\cite{And2001.PRL86.2926,Dut2001.Sci293.663,zwierlein_ring}.
        The topic of one~\cite{shuang} and few vortices (possibly
        of different signs~\cite{chops,chops1,simula1,simula2})
        remains under active experimental investigation still to this day.

        The study of BECs at the mean-field level at (and very close to) zero
        temperatures is well established to be described by the
        famous Gross-Pitaesvkii (GP) partial differential equation
        (PDE)~\cite{becbook1,becbook2,siambook}. When rewriting
        the equation for the complex wavefunction into a system of
        equations for the density $\rho$ and velocity $v=\nabla \phi$
        (where $\phi$ is the phase of the complex field) one obtains a system
        strongly reminiscent of the Euler equations in fluid dynamics.
        For a recent, detailed account of this connection, see,
        e.g.,~\cite{bradley}. The role of the quantum features arises
        through the so-called quantum pressure term. Nevertheless,
        this analogy can be utilized to approximate the vortex
        dynamics and interactions within the GP system by those of
        point vortices in the fluid setting; for a recent discussion
        of how to utilize configurations of the latter to prove
        the existence of steady or co-traveling states in the former,
        see, e.g.,~\cite{wei}. There is a time-honored history
        of connections between the theory of different types
        of polynomials and the study of vortices in
        fluids~\cite{kadtke,aref1,aref2}. Recent years have seen
        an attempt to extend such considerations to the more complex
        (in that it bears an external trapping potential) setting
        of BECs, including extensions of relevant multi-vortex
        configurations~\cite{barry1} and associated polynomial generating
        function techniques~\cite{barry2}. 
         It is along these lines of associating the equations for the vortex positions (and their conjugates) with a system of polynomial equations and using symbolic algebraic techniques to tackle the latter that we proceed in the present study.

         Solving polynomial equations is central to algebraic geometry. Over
         the past century many approaches have been developed, and with them,
         many algorithms.  Fifty years ago, Buchberger proposed the Gr\"obner
         basis algorithm for solving polynomial systems \cite{Buchberger},
         which was improved by Faug\`ere \cite{F5},
         and
        is now the  central tool in computational algebraic geometry (see, for example, \cite{CoxLittleOShea}). Gr\"obner bases are  implemented in several packages, from general symbolic software such as Mathematica \cite{Mathematica} and Maple \cite{Maple}, to specialised software such as Magma \cite{Magma}, Singular \cite{Singular} and Macaulay2 \cite{M2}. Despite the computational complexity of Gr\"obner bases being double exponential in the worst case, there has been much success in using Gr\"obner basis techniques to solve problems in a wide range of applications including shear flow, nonlinear mechanics, chemical reactions, dynamical systems, statistics, systems biology and computer vision, among many others \cite{pausch,liu,drton,gatermann,sadeghimanesh,laubenbacher,harrington2012,aholt}.
 
Our aim here is to use such methods to \emph{fully} characterize all possible solutions of the system of vortex equations for stationary configurations involving up to 5 vortices and offer all the computationally tractable solutions for 6 vortices.
There are two natural avenues for obtaining polynomial equations to describe the vortex positions.
We either consider the conjugate variables of the vortex positions
or separate the equations in terms of real and imaginary parts. However, solving these systems using standard Gr\"obner basis approaches is computationally out of reach (see Table \ref{table:benchmark}). The next step is then
to attempt to exploit any symmetries in the equations. The system of equations for the vortex positions (see Equation \ref{Eq:Zeroth}) is symmetric in the variables, that is, it is invariant under the action of a permutation group. However individual equations in the system are not invariant. If the system was given by algebraic equations, then theoretically there is a system consisting of invariant equations which has the same solutions and can be rewritten in terms of a basic set of invariants (see \cite{SturmfelsIT, DerksenKemper}). However using the standard algorithms in invariant theory is, again, computationally intractable for the vortex equations.

While all the general approaches fail for the vortex problem, one can tailor ideas from computational algebra and invariant theory to gain new insight about vortex configurations. For example, Faug\`ere and Svartz \cite{FaugSvar}
proposed a general approach for systems with an action of the permutation group and they specialised their approach to find configurations of up to 8 single charged vortices. 
Here we extend this analysis to the equally relevant
case of vortices with opposite charges.
While vortex configurations of the same charge are rigidly rotating, those
with opposite charges can be genuinely stationary~\cite{siambook}.
For the latter, we specifically first construct a family of symmetric equations in the variables and their conjugates whose set of common zeros includes the solutions to the ODE approximation of the vortex problem. Next, from these equations we deduce a set of equations in the elementary symmetric polynomials in the variables. The common zero set of these equations, again, includes the solutions to the ODE approximation of the vortex problem. 
Finally we solve the system in the elementary symmetric polynomials and we convert the solutions into the original variables. 
Practical considerations presently prohibit an exhaustive study of larger
(than $6$, as considered herein) numbers of vortices, although this is an interesting
and practically important topic for further study.
Once we have these solutions, our final step is to come full circle
and check the existence of such configurations in the original PDE
problem. This is explored in some detail
(and at different levels of approximation) in Section~\ref{Section:BacktoPDE}.

Our paper is organised as follows. In Section \ref{Section:VorticesToAlgebra} we explain how to translate the underlying PDE to polynomial systems. The theoretical results and proofs of the two charged vortex problem using algebraic approaches is presented in Section \ref{Section:ExploitingSymmetries}. In Section \ref{Section:Results} we give the results from the computational algebra for the configurations of vortices and the benchmarking of the various approaches, showing that the direct approaches fail. In Section~\ref{Section:BacktoPDE}, we show that numerical solutions of the PDE corroborate the presently identified
solution, as well as all the
ones for smaller vortex numbers.
Finally, in the last section, we give conclusions and suggest some directions
for future work.

%%%%%%%%%%%%%%%%%%%%%%%%%%%%%%%%%%%%%%%%%%%%%%%
%
%       SECTION: Polynomial equations for the vortex problem
%
%%%%%%%%%%%%%%%%%%%%%%%%%%%%%%%%%%%%%%%%%%%%%%%

	\section{From vortices to algebra}\label{Section:VorticesToAlgebra}

        The original PDE problem involves identifying standing
wave solutions of the two-dimensional Gross--Pitaevskii equation of the
form~\cite{becbook1,becbook2}:
\begin{eqnarray}
i u_{t} = - \frac{1}{2} \Delta u + |u|^2 u + V u,
\label{GPE}
\end{eqnarray}
where $\Delta$ represents the two-dimensional Laplacian and the 
external trapping potential $V(r)$ is assumed to be parabolic, i.e., $V(r)=\frac{1}{2}
\Omega^2 (r)$, where $r=x^2+y^2$ and $\Omega$ is the trap. 
Notice that this dimensionless version of the GP equation
has been obtained through well-established reductions of
the dimensional one, as detailed, e.g., in~\cite{siambook}.
We use solutions of the form $u(x,y,t)=e^{-i \mu t}
U(x,y)$ with chemical potential
$\mu$ and subsequently solve the nonlinear steady state problem for
$U(x,y)$ via a fixed point iteration.

When identifying solutions bearing vortices, we can attempt
to capture the effective dynamics of the vortices
through the following ordinary differential equations:
\begin{eqnarray}
i \dot{z}_j=S_j \omega_{pr} z_j + \sum_{1\leq k\neq j \leq n} \frac{S_k}{\bar{z}_j - 
\bar{z}_k.}
\label{aux1}
\end{eqnarray}
Here $S_j$ are the charges of the vortices, $\omega_{pr}$ is their
precession frequency  inside the parabolic trap, $i$ is $\sqrt{-1}$ and the complex number
$z_j$ represents the planar position
of the $j$-the vortex. Near the center of the trap it is reasonable
to assume that vortices have a nearly constant precession frequency
$\omega_{pr}$.
This has been described, e.g., recently in~\cite{shuang}.
The second term captures the inter-vortex interactions of the $j$-th
vortex with all the other vortices (summed over $k$). Notice that
this is a velocity-induced interaction, i.e., each vortex induces
a velocity field at the location of all the others, as is the
case for vortices in inviscid, incompressible fluid (point) vortices. 
Notice that this description is progressively more accurate for
larger values of the chemical potential $\mu$, whereby the vortices
approach the limit of point vortices with (decreasing
width and thus) no internal structure.

However, the ground state of the system, i.e., the background over
which the vortices are located decreases in its density as one
moves radially outwards, due to the presence of
the parabolically confining external potential.
As has been extensively examined since early on~\cite{fetter}
(see also the recent discussion of~\cite{shuang}),
this has an implication of radially increasing
the precession frequency according to
\begin{eqnarray}
  \omega_{pr}(|z_j|) = \frac{\omega_{pr}(0)}{1-\frac{V(|z_j|)}{\mu}}.
  \label{aux2}
\end{eqnarray}
In turn, this leads to the amended version of the equations
of motion as:
\begin{eqnarray}
i \dot{z}_j= S_j \frac{\omega_{pr}(0)}{1-\frac{V(|z_j|)}{\mu}} z_j
+ \sum_{1\leq k\neq j \leq n} \frac{S_k}{\bar{z}_j-\bar{z_k}}.
\label{aux3}
\end{eqnarray}

However, it turns out that the interactions between the vortices
are also affected by the presence of the external potential.
In particular, the interactions between the vortices as characterized
in Equation~(\ref{aux1}) assume the presence of {\it a homogeneous background}
in which the vortices move. A spatially inhomogeneous background,
present in the case of the trap, modifies (i.e., {\it screens}) the
inter-vortex interactions in a way that has been recently captured,
e.g., in~\cite{ourpra2}. This, in turn, as discussed in this work (see
also references therein),
can be captured by a modulating factor, leading to a further revised
form of the equations. The latter account for the inhomogeneous
background in both the individual vortex precession and
in the inter-vortex interactions and read:
\begin{eqnarray}
  i \dot{z}_j= S_j \frac{\omega_{pr}(0)}{1-\frac{V(|z_j|)}{\mu}} z_j
  + \sum_{1\leq k\neq j \leq n} \frac{V(|z_j|)}{V(|z_k|)} \frac{S_k}{\bar{z}_j-\bar{z_k}}.
  \label{aux4}
\end{eqnarray}

Our aim in the present work is to explore equilibrium
multi-vortex configurations involving vortices
of both charges, i.e., with $S_j=\pm 1$. In particular,
we will more specifically assume that we have
$M$ vortices with $S_j=1$ and $N$ vortices with $S_j=-1$.
We will denote the former with the positions 
$x_1, \ldots, x_M$ in the complex plane $\IC$
and the latter with the positions $y_1, \ldots, y_N$
in the same plane. Our analysis will take place at
the level of the simplest
equation for the vortices, namely Equation~(\ref{aux1}).
However, in Section~\ref{Section:BacktoPDE}, we will illustrate
the connections of this setting with the full original
problem, as well as the more elaborate (and
more accurate) variations of the form of Equation~(\ref{aux3})
and (\ref{aux4}). At the level of Equation~(\ref{aux1}),
splitting the equations for positive and negative
charges, according to the symbolism above, we obtain
the steady state formulation:

\begin{eqnarray}\label{Eq:Zeroth}
  \left\{
  \begin{aligned}
	\overline{x_i} = - \sum_{\substack{j \neq i \\ j=1}}^M \frac{1}{x_i-x_j} + \sum _{j=1}^N \frac{1}{x_i-y_j}, \ \ \ i = 1,\ldots M,\\
	\overline{y_i} = - \sum_{\substack{j \neq i \\ j=1}}^N \frac{1}{y_i-y_j} + \sum _{j=1}^M \frac{1}{y_i-x_j}, \ \ \ i = 1,\ldots N.
  \end{aligned}
  \right.
\end{eqnarray}

Note that the solutions of this system present many symmetries:

	\begin{lemma}[Symmetries of the vortex equations]\label{lemma:symmetries}
The set of solutions to the system~\eqref{Eq:Zeroth} is invariant under up to four different group actions.
\begin{itemize}
	\item The product of symmetric groups $S_M\times S_N$ acts by permuting the variables via $(\sigma_1, \sigma_2)\cdot x_i = x_{\sigma_1(i)}$, $(\sigma_1, \sigma_2)\cdot y_j = y_{\sigma_2(j)}$. 

	\item The group of complex numbers of modulus $1$ acts via $\lambda\cdot x_i=\lambda x_i$ and $\lambda\cdot y_i=\lambda y_i$. This corresponds to rotations of the complex plane around the origin.
	\item The cyclic group of order 2 acts via conjugation. This correspond to reflection of the complex plane with respect to the real axis.
	\item When $M=N$, the cyclic group of order 2 acts by exchanging $x_i$ and $y_i$.
\end{itemize}
\end{lemma}
\begin{proof}
This is verified via straightforward computations.
\end{proof}

%---------------------------------------------------------------------------------------------------------------------

%%%%%%%%%%%%%%%%%%%%%%%%%%%%%%%%%%%%%%%%%%%%%%%
%
%       SECTION: Finding solutions: direct aproaches
%
%%%%%%%%%%%%%%%%%%%%%%%%%%%%%%%%%%%%%%%%%%%%%%%

\subsection{Obtaining algebraic equations: direct approaches } \label{Section:direct}

System~\eqref{Eq:Zeroth} fails to be algebraic because of the presence of the conjugation operator. In the next two subsections, we discuss two direct approaches for converting this problem into an algebraic problem, which provides the possibility to use algebraic methods. Both approaches end up doubling the number of variables and equations. We show in Section~\ref{Section:Results} that these approaches cannot succeed in their practical implementation, unless
the resulting system is not suitably reduced.
%---------------------------------------------------------------------------------------------------------------------
%---------------------------------------------------------------------------------------------------------------------
\subsubsection{Conjugate variables}

 A first method of converting system~\eqref{Eq:Zeroth} to an algebraic system is to introduce a new sets of variables $X_1,\ldots X_M, Y_1,\ldots Y_N$ representing the complex conjugates of the variables
 $x_1,\ldots x_M, y_1,\ldots y_N$. Therefore we obtain the following equations:

\begin{align*}
	X_i = - \sum_{\substack{j \neq i \\ j=1}}^M \frac{1}{x_i-x_j} + \sum _{j=1}^N \frac{1}{x_i-y_j}, \ \ \ i = 1,\ldots M, \tag{$E^{(x)}_i$}\label{Eq:Exi} \\
	Y_i = - \sum_{\substack{j \neq i \\ j=1}}^N \frac{1}{y_i-y_j} + \sum _{j=1}^M \frac{1}{y_i-x_j}, \ \ \ i = 1,\ldots N, \tag{$E^{(y)}_i$}\label{Eq:Eyi} \\
	x_i = - \sum_{\substack{j \neq i \\ j=1}}^M \frac{1}{X_i-X_j} + \sum _{j=1}^N \frac{1}{X_i-Y_j}, \ \ \ i = 1,\ldots M, \tag{$\bar{E}^{(x)}_i$}\label{Eq:BExi} \\
	y_i = - \sum_{\substack{j \neq i \\ j=1}}^N \frac{1}{Y_i-Y_j} + \sum _{j=1}^M \frac{1}{Y_i-X_j}, \ \ \ i = 1,\ldots N. \tag{$\bar{E}^{(y)}_i$}\label{Eq:BEyi}
\end{align*}

We now have $2(M+N)$ equations in $2(M+N)$ variables. 

\begin{remark}\label{rmk:symmetriesConj}
Each of the symmetries described in Lemma \ref{lemma:symmetries} extend to the solutions of this new system naturally.
\end{remark}

Polynomial equations for this system are obtained from these equations by clearing the denominators.
 Let us denote by $D$ the discriminant $D = \prod_{i \neq j} (x_i-x_j) \cdot \prod_{i, j} (x_i-y_j) \cdot \prod_{i \neq j} (y_i-y_j)$ and by $\bar{D}$ the discriminant in the conjugate variables $\bar{D} = \prod_{i \neq j} (X_i-X_j) \cdot \prod_{i, j} (X_i-Y_j) \cdot \prod_{i \neq j} (Y_i-Y_j)$.
 Since we cleared the denominators, we need to exclude solutions where $D$ and $\bar{D}$ are zero. This is done by introducing two new variables $h$ and $H$ and by adding to the system the equations $h \cdot D-1=0$ and $H\cdot \bar{D} -1=0$.
 Finally, to reduce the dimension of the set of solutions, we need to remove the symmetry given by the multiplicative action of complex numbers of modulus $1$ mentioned in Remark \ref{rmk:symmetriesConj}. To do so, we impose the condition that some non-zero variable of the system is real.
We do this in the following way:
\begin{enumerate}
\item we subdivide the system into two subsystems corresponding to the cases $x_1 = 0$ and $x_1 \neq 0$. In particular, the first subsystem is obtained by adding the new equation $x_1=0$, the second subsystem is obtained by adding a new variable $a$ and the new equation $a\cdot x_1 - 1=0$,   
	\item we add the new equation $y_1-Y_1=0$ to the subsystem where $x_1=0$ (i.e. we assume $y_1\neq x_1=0$ is real) and we add the new equation $x_1-X_1=0$ to the subsystem where $x_1 \neq 0$ (i.e. we assume $x_1\neq 0$ is real).
\end{enumerate}
The two systems so obtained can be solved using standard algebraic geometric tools and in the cases considered yield finitely many solutions. To get the solutions of system~\eqref{Eq:Zeroth} we now just need to select those solutions where $x_i = \bar{X}_i$ and $y_i = \bar{Y}_i$.

\subsubsection{Real and imaginary parts}

An alternative approach to describe system~\eqref{Eq:Zeroth} as a polynomial system consists in doubling up the number of equations and variables by considering real and imaginary parts separately. We introduce $2M+2N$ variables $a_1, \ldots, a_M, b_1, \ldots, b_M, c_1, \ldots, c_N, d_1, \ldots d_N$. We substitute the variables in system~\eqref{Eq:Zeroth} via $x_i =  a_i + \sqrt{-1} b_i$ and $y_i =  c_i + \sqrt{-1} d_i$ and we clear the denominators. Separating the real and
imaginary parts of the resulting $M+N$ equations, we get a system of $2M+2N$ real polynomial equations, two for each equation of system~\eqref{Eq:Zeroth}. Finding solutions to system~\eqref{Eq:Zeroth} is now equivalent to finding real solutions to this new system. The discriminant $D$ can be written as $D = \prod_{i \neq j} (a_i + \sqrt{-1} b_i - a_j - \sqrt{-1}b_j) \cdot \prod_{i, j} (a_i + \sqrt{-1}b_i - c_j - \sqrt{-1}d_j) \cdot \prod_{i \neq j} (c_i + \sqrt{-1}d_i - c_j - \sqrt{-1}d_j)$.
The discriminant $D$ can be written as a sum $D = D_{\text{re}} + \sqrt{-1} D_{\text{im}}$, where $ D_{\text{re}}$ and $D_{\text{im}}$ are real polynomials.
We now add to the systems an extra variable $h$ and the equation $h\cdot  D_{\text{re}} \cdot D_{\text{im}} -1=0$ to ensure that the real and the imaginary part of $D$ are not both zero. Here, the multiplicative action of complex numbers of modulus $1$ becomes the group of rotations around the origin in the real plane. We get rid of this 1-dimensional symmetry by requiring that one of $x_1=a_1+\sqrt{-1}b_1$ or $y_1=c_1+\sqrt{-1}d_1$ is purely imaginary:
\begin{enumerate}
	\item we subdivide the system into two subsystems corresponding to the cases $a_1 + \sqrt{-1} b_1 = 0$ and $a_1 + \sqrt{-1} b_1 \neq 0$. In particular, the first subsystem is obtained by adding the new equations $a_1=0$ and $b_1=0$, the second subsystem is obtained by adding a new variable $k$ and the new equation $k\cdot a_1 \cdot b_1 - 1=0$, 
	\item we add the new equation $c_1=0$ to the subsystem where $a_1 + \sqrt{-1} b_1 = 0$ and we add the new equation $a_1=0$ to the subsystem where  $a_1 + \sqrt{-1} b_1 \neq 0$.
\end{enumerate}

 Real solutions of this system can now be obtained either by finding all complex solutions with symbolic methods and then picking out the real solutions, or by approximation with numerical algebraic geometry (see Section \ref{Section:Results}).

%%%%%%%%%%%%%%%%%%%%%%%%%%%%%%%%%%%%%%%%%%%%%%%
%
%       SECTION: Finding solutions: exploiting the symmetries
%
%%%%%%%%%%%%%%%%%%%%%%%%%%%%%%%%%%%%%%%%%%%%%%%

\section{Obtaining algebraic equations: exploiting the symmetries}\label{Section:ExploitingSymmetries}

In this section we exploit the invariance of system~\eqref{Eq:OriginalSystem} under the action of $S_M\times S_N$ to obtain a new set of equations. The equations we obtain are not simpler to the eye, far from it, but will prove to be better for symbolic computations (See Section \ref{Section:Results}).

In the rest of this section we explain how to exploit the symmetries described in Lemma  \ref{lemma:symmetries} and Remark \ref{rmk:symmetriesConj} in order to go further with computations than what can be achieved with the direct approaches. We focus on the action of the product of symmetric groups. Our starting point is the conjugate variables system. Invariant theory of finite groups suggests
that since the system is invariant there exists a set of invariant equations which have the same set of common zeros \cite{SturmfelsIT}. Furthermore these invariant equations can be written as polynomials in a finite generating set of the polynomial invariants \cite{DerksenKemper}. The idea is then to solve for the value of these generators, with the hope that this computation is more feasible than the direct computation. There are two main problems with this plan. First, the symbolic methods for performing this ``symmetrization'' and ``rewriting'' rely on Gr\"obner bases and so are quickly computationally intractable.
Second, the set of generating invariants for the action on the variables $x_i,X_i,y_i,Y_i$ is complicated. If we consider only the action on the variables $x_i,y_i$, then a generating set is simple enough, indeed one may take the elementary symmetric polynomials $e^{(x)}_1, \ldots,  e^{(x)}_M$ in the variables $x_1, \ldots, x_M$ and the  elementary symmetric polynomials $e^{(y)}_1, \ldots,  e^{(y)}_N$  in the variables $y_1, \ldots, y_N$, given by $e^{(x)}_k:= \sum_{I\subseteq [M], |I|=k}\prod_{i\in I} x_i$ and $e^{(y)}_k:= \sum_{I\subseteq [N], |I|=k}\prod_{i\in I} y_i$, respectively.

We take an indirect approach inpired by the methodology
utilized by Faug\`ere and Svartz towards solving for the (rigidly rotating)
configurations
of vortices of a single charge \cite{FaugSvar}. The first step is to construct a set of invariant equations, written in terms of invariants in the variables $x_i,X_i,y_i,Y_i$, whose set of common zeros includes the solutions to System \eqref{Eq:Zeroth} (i.e. those common zeros of the conjugate variable system such that $X_i=\overline{x_i}$ and $Y_i=\overline{y_i}$). From these, we then deduce symmetric equations in the $x_i,y_i$ whose set of common zeros includes
the solutions to System \eqref{Eq:Zeroth}. The rewriting is done at the same time. The section ends with an explanation of our solution procedure.

%---------------------------------------------------------------------------------------------------------------------

			\subsection{Invariant equations}\label{SSec:InvEq}

                        In this section we introduce some
                     invariant equations that are satisfied by the solutions to the vortex problem. 

To start, we provide a useful compact form of the conjugate system. We set $P(z) = \prod_{i = 1}^M (z-x_i)$ and $Q(z) = \prod_{i = 1}^N (z-y_i)$.
We have the following Lemma.

\begin{lemma}
The system $\{E^{(x)}_i, \bar{E}^{(x)}_i, E^{(y)}_j, \bar{E}^{(y)}_j \mid i=1, \ldots, M, j=1\ldots N\} $ is equivalent to:
\begin{eqnarray}\label{Eq:OriginalSystem}
  \left\{
  \begin{aligned}
	X_i =  - \frac{P''(x_i)}{2P'(x_i)} + \frac{Q'(x_i)}{Q(x_i)},\ \ x_i =  - \frac{P''(X_i)}{2P'(X_i)} + \frac{Q'(X_i)}{Q(X_i)}, \ \ \ i = 1,\ldots M,\\
	Y_i = - \frac{Q''(y_i)}{2Q'(y_i)} + \frac{P'(y_i)}{P(y_i)},\ \ \ y_i = - \frac{Q''(Y_i)}{2Q'(Y_i)} + \frac{P'(Y_i)}{P(Y_i)}, \ \ \ i = 1,\ldots N.
  \end{aligned}
  \right.
\end{eqnarray}
\proof
The proof of \cite[Lemma~1]{FaugSvar} shows that $\sum_{\substack{j\neq i \\ j=1}}^M \frac{1}{x_i-x_j} = \frac{P''(x_i)}{2P'(x_i)}$. 
Combining with $\frac{Q'(z)}{Q(z)} = \sum_{j=1}^N \frac{1}{z-y_j}$ we get that equation $E^{(x)}_i$ can be rewritten as $X_i =  - \frac{P''(x_i)}{2P'(x_i)} + \frac{Q'(x_i)}{Q(x_i)}$. Similar computations hold for the other equations. 
\endproof
\end{lemma}

We denote by $s^{(x)}_k$ the Newton sum $s^{(x)}_k = \sum_{i=1}^M x_i^k$, and we define $r^{(x)}_k = \sum_{i=1}^M x_i^kX_i$. Similarly,  we denote $s^{(y)}_k = \sum_{i=1}^M y_i^k$, and $r^{(y)}_k = \sum_{i=1}^M y_i^kY_i$.

\begin{theorem}\label{thm:invariant1}
For every $k\geq 0$, the solutions to System \eqref{Eq:Zeroth} satisfy the equation
\begin{align*}
r_k&^{(x)} +r_k^{(y)} =\\
&  - \frac{1}{2} \left( \sum_{i=0}^{k-1}s_i^{(x)}s_{k-i-1}^{(x)}\right) + \frac{k}{2} s_{k-1}^{(x)} - \frac{1}{2} \left( \sum_{i=0}^{k-1}s_i^{(y)}s_{k-i-1}^{(y)}\right) + \frac{k}{2} s_{k-1}^{(y)} + \left( \sum_{i=0}^{k-1}s_{i}^{(x)}s_{k-i-1}^{(y)}\right).
\end{align*}
\end{theorem}
\proof
We have 

\begin{equation}
r_k^{(x)} = \sum_{i=1}^M x_i^kX_i = -\sum_{i=1}^M \frac{x_i^kP''(x_i)}{2P'(x_i)} + \sum_{i=1}^M \frac{x_i^k Q'(x_i)}{Q(x_i)}.
\end{equation}

By the proof of \cite[Theorem~4]{FaugSvar} we have $-\sum_{i=1}^M \frac{x_i^kP''(x_i)}{2P'(x_i)} = - \frac{1}{2} \left(\sum_{i=0}^{k-1}s_i^{(x)}s_{k-i-1}^{(x)}\right) + \frac{k}{2} s_{k-1}^{(x)}$. As a  result,

\begin{equation}\label{Eq:rkx}
r_k^{(x)} =  - \frac{1}{2} \left(\sum_{i=0}^{k-1}s_i^{(x)}s_{k-i-1}^{(x)}\right) + \frac{k}{2} s_{k-1}^{(x)} + \sum_{i=1}^M \sum_{j=1}^N \frac{x_i^k}{x_i-y_j}.
\end{equation}

Similarly we have

\begin{equation}\label{Eq:rky}
r_k^{(y)} = - \frac{1}{2} \left(\sum_{i=0}^{k-1}s_i^{(y)}s_{k-i-1}^{(y)}\right) + \frac{k}{2} s_{k-1}^{(y)} + \sum_{i=1}^M \sum_{j=1}^N \frac{-y_j^k}{x_i-y_j}.
\end{equation}

If we sum Equation~\eqref{Eq:rkx} and Equation~\eqref{Eq:rky} we obtain

\begin{align}\label{Eq:rky0}
r_k&^{(x)} +r_k^{(y)} = \\
\nonumber & - \frac{1}{2} \left( \sum_{i=0}^{k-1}s_i^{(x)}s_{k-i-1}^{(x)} \right) + \frac{k}{2} s_{k-1}^{(x)} - \frac{1}{2} \left( \sum_{i=0}^{k-1}s_i^{(y)}s_{k-i-1}^{(y)} \right) + \frac{k}{2} s_{k-1}^{(y)} + \sum_{i=1}^M \sum_{j=1}^N \frac{x_i^k-y_j^k}{x_i-y_j}.
\end{align}

Let us assume $k\geq 1$. We have $\frac{x_i^k-y_j^k}{x_i-y_j} = \sum_{m=0}^{k-1} x_i^my_j^{k-m-1}$.
It follows that 

\begin{equation}\label{Eq:rky1}
\sum_{i=1}^M \sum_{j=1}^N \frac{x_i^k-y_j^k}{x_i-y_j} = \sum_{i=1}^M \sum_{j=1}^N \sum_{m=0}^{k-1} x_i^my_j^{k-m-1} =\sum_{m=0}^{k-1}s_{m}^{(x)}s_{k-m-1}^{(y)}.
\end{equation}

Therefore, for $k\geq 1$, the statement follows immediately from Equation~\eqref{Eq:rky0} and Equation~\eqref{Eq:rky1}.

Finally, we consider the case $k=0$. We have

\begin{align}
r_0^{(x)} &= \sum_{i=1}^M X_i = - \sum_{i=1}^M \sum_{\substack{j=1 \\ j\neq 1}}^M \frac{1}{x_i-x_j} + \sum_{i=1}^N \sum_{j=1}^M \frac{1}{x_i-y_j} \\
 & =  - \sum_{\substack{i,j=1 \\ i<j}}^M \frac{1}{x_i-x_j}  -  \sum_{\substack{i,j=1 \\ i<j}}^M \frac{1}{x_j-x_i} +\sum_{i=1}^N \sum_{j=1}^M \frac{1}{x_i-y_j}  = \sum_{i=1}^N \sum_{j=1}^M \frac{1}{x_i-y_j}, \\
r_0^{(y)} &= \sum_{i=1}^N \sum_{j=1}^M \frac{1}{y_j-x_i},
\end{align}

and therefore $r_0^{(x)} + r_0^{(y)} = 0$ as desired.
\endproof

%---------------------------------------------------------------------------------------------------------------------

				\subsection{Equations in the elementary symmetric functions}\label{SSec:InvFun}

In Theorem~\ref{thm:invariant1} we introduced a set of invariant equations to describe the vortex problem. For practical uses, the presence of $r_k^{(x)}$ and $r_k^{(y)}$ in these equations produces two disadvantages:
\begin{enumerate}
	\item they involve both the variables $x_i$'s, $y_j$'s and their conjugates $X_i$'s, $Y_j$'s,
	\item there is no easy formula to express $r_k^{(x)}$ in terms of the elementary symmetric functions in the $x_i$'s and $X_i$'s.
\end{enumerate}
We remind the reader that, on the other hand, there exist well known formulas to express the Newton sums $s^{(x)}_k$'s in terms of the elementary symmetric functions $e^{(x)}_k$'s.
In this Section we obtain a new set of invariant equations for the vortex problem that avoid these issues. 

 Given a polynomial $F\in K[e_1^{(x)}, \ldots e_M^{(x)}, e_1^{(y)}\ldots e_N^{(y)}][z]$ in the symmetric polynomials $e^{(x)}_i$'s and $e^{(y)}_i$'s
and in the extra variable $z$, we wish to express in a compact way the sum $\sum_{i=1}^M F(x_i) \pm \sum_{j=1}^N F(y_j)$. To do so, we consider the transformations $\mathscr{S^+}$ and $\mathscr{S^-}$ defined by

\begin{equation*}
\begin{array}{ccccc}
\mathscr{S^+}: & K[e_1^{(x)}, \ldots e_M^{(x)}, e_1^{(y)}\ldots e_N^{(y)}][z] & \rightarrow & K[e_1^{(x)}, \ldots e_M^{(x)}, e_1^{(y)}\ldots e_N^{(y)}] \\
 & \sum a_k z^k & \mapsto & \sum a_k (s_k^{(x)} + s_k^{(y)}),
\end{array}
\end{equation*}
and
\begin{equation*}
\begin{array}{ccccc}
\mathscr{S^-}: & K[e_1^{(x)}, \ldots e_M^{(x)}, e_1^{(y)}\ldots e_N^{(y)}][z] & \rightarrow & K[e_1^{(x)}, \ldots e_M^{(x)}, e_1^{(y)}\ldots e_N^{(y)}] \\
 & \sum a_k z^k & \mapsto & \sum a_k (s_k^{(x)} - s_k^{(y)}),
\end{array}
\end{equation*}

where the $a_k$'s are polynomials not involving the variable $z$.
In other words, $\mathscr{S^\pm}$ acts by expanding $F$ in the variable $z$ and then replacing the power  $z^k$ with the expression $(s_k^{(x)} \pm s_k^{(y)})$, which, we remind, can be expressed in terms of the symmetric polynomials $e^{(x)}_i$'s and $e^{(y)}_i$'s.

Let $F,G \in K[e_1^{(x)}, \ldots e_M^{(x)}, e_1^{(y)}\ldots e_N^{(y)}][z]$ be two polynomials, and write them as $F = \sum f_kz^k$ and $G = \sum g_kz^k$. Then, their sum $F+G$ can be written as $F+G = \sum (f_k+g_k)z^k$. It follows that  $\mathscr{S^\pm}(F+G) = \mathscr{S^\pm}(F) + \mathscr{S^\pm}(G)$. 
Similarly, given $F \in K[e_1^{(x)}, \ldots e_M^{(x)}, e_1^{(y)}\ldots e_N^{(y)}][z]$ and a polynomial $h \in K[e_1^{(x)}, \ldots e_M^{(x)}, e_1^{(y)}\ldots e_N^{(y)}]$ not involving the variable $z$, we have $\mathscr{S^\pm}(hF) = h \mathscr{S^\pm}(G)$.
However, in general, for $F,G \in K[e_1^{(x)}, \ldots e_M^{(x)}, e_1^{(y)}\ldots e_N^{(y)}][z]$, we have  $\mathscr{S^\pm}(FG) \neq \mathscr{S^\pm}(F) \mathscr{S^\pm}(G)$.
These properties sum up to say, in the language of commutative algebra, that $\mathscr{S^+}$ and $\mathscr{S^+}$ are morphisms of $K[e_1^{(x)}, \ldots e_M^{(x)}, e_1^{(y)}\ldots e_N^{(y)}]$-modules, but not homorphisms of rings.
As mentioned above, the transformations $\mathscr{S^+}$ and $\mathscr{S^-}$ allow to write in a convenient way sums of the type $\sum_{i=1}^M F(x_i) \pm \sum_{j=1}^N F(y_j)$. 

\begin{lemma}
For every $F \in K[e_1^{(x)}, \ldots e_M^{(x)}, e_1^{(y)}\ldots e_N^{(y)}][z]$ we have 
\begin{equation*}
	\mathscr{S}^\pm(F) = \sum_{i=1}^M F(x_i) \pm \sum_{j=1}^N F(y_j).
\end{equation*}
\proof
We write $F(z) = \sum_k a_k z^k$. We have 
\begin{multline}
\mathscr{S}^\pm(F) = \sum_k a_k s^{(x)}_k \pm \sum_k a_k s^{(y)}_k = \sum_k \sum_{i=1}^M a_k x_i^k \pm \sum_k \sum_{j=1}^N a_k y_j^k  \\ 
 = \sum_{i=1}^M \sum_k a_k x_i^k \pm \sum_{j=1}^N \sum_k  a_k y_j^k =  \sum_{i=1}^M F(x_i) \pm \sum_{j=1}^N F(y_j).
\end{multline}
\endproof
\end{lemma}

\begin{corollary}
For every $F \in K[e_1^{(x)}, \ldots e_M^{(x)}, e_1^{(y)}\ldots e_N^{(y)}][z]$ we have $\mathscr{S}^\pm(FPQ) = 0$.
\proof
We have $\mathscr{S}^\pm(FPQ) = \sum_{i=1}^M F(x_i)P(x_i)Q(x_i) \pm \sum_{j=1}^N F(y_j)Q(y_j)P(y_j) = 0$,
since $P(x_i)=Q(y_j)=0$ for every $i=1,\ldots, M$ and $j=1,\ldots,N$.
\endproof
\end{corollary}

We are now ready to write expressions for  $S^{(x)}_k$, $S^{(y)}_k$, $R^{(x)}_k$ and $R^{(y)}_k$ in terms of $\mathscr{S}^+$ and $\mathscr{S}^-$. 

Let $D$ be the resultant of $PQ$ and $(PQ)'$. By definition the \emph{resultant of the polynomials $PQ$ and $(PQ)'$ in one variable $z$} is a polynomial in their coefficients that vanishes if and only if $PQ$ and $(PQ)'$ have a
common root. In particular, it is an element of $K[e_1^{(x)}, \ldots e_M^{(x)}, e_1^{(y)}\ldots e_N^{(y)}]$, and it can be explicitly computed as the determinant of the Sylvester matrix. 

\begin{remark}
The resultant of $PQ$ and $(PQ)'$ is the expression in terms of the elementary symmetric polynomials of the discriminant introduced in Section \ref{Section:VorticesToAlgebra}. Indeed, $PQ$ and $(PQ)'$ have a common root if and only if $PQ$ has a double root, which happens only if two among $x_1, \ldots, x_M, y_1,\ldots, y_N$ coincide.
\end{remark}

The discriminant $D$
also satisfies the equation
\begin{align}
B (z) P(z)Q(z) &+ C(z) (P(z)Q(z))' \\ 
\nonumber &= B(z) P(z)Q(z) + C(z) (P'(z)Q(z) + P(z)Q'(z)) = D,
\end{align}
for some $B, C \in K[e_1^{(x)}, \ldots e_M^{(x)}, e_1^{(y)}\ldots e_N^{(y)}][z]$.

We then have
\begin{equation}
S^{(x)}_k = \sum_{i=1}^M X_i^k = \sum_{i=1}^M \left(\frac{-P''(x_i)Q(x_i) + 2P'(x_i)Q'(x_i)}{2P'(x_i)Q(x_i)}\right)^k. 
\end{equation}
Denote 
\begin{equation}
A(z) \eeq \frac{1}{2} C (z)( -P''(z)Q(z) + 2P'(z)Q'(z) - P(z)Q''(z)).
\end{equation}
We have 
\begin{equation}
S^{(x)}_k = \sum_{i=1}^M \left( \frac{A(x_i)}{C(x_i)P'(x_i)Q(x_i)} \right)^k = \sum_{i=1}^M \left( \frac{A(x_i)}{C(x_i)(P'(x_i)Q(x_i)+P(x_i)Q'(x_i))} \right)^k, 
\end{equation}
where we are using on both numerator and denominator the fact that $P(x_i) = 0$ for $i=1, \ldots, M$.

Similarly, for $S^{(y)}$ we get the expression
\begin{equation}
S^{(y)}_k = \sum_{j=1}^N \left( \frac{A(y_j)}{C(y_j)P(y_j)Q'(y_j)} \right)^k = \sum_{j=1}^N \left( \frac{A(y_j)}{C(y_j)(P'(y_j)Q(y_j)+P(y_j)Q'(y_j))} \right)^k.
\end{equation}

So, now

\begin{align}
S^{(x)}_k \pm S^{(y)}_k &= 
\sum_{i=1}^M \left( \frac{A(x_i)}{D - B(x_i) P(x_i) Q(x_i)} \right)^k \pm\sum_{j=1}^N \left( \frac{A(y_j)}{D - B(y_j) P(y_j) Q(y_j)} \right)^k  \\
\nonumber &=\sum_{i=1}^M \left( \frac{A(x_i)}{D} \right)^k \pm \sum_{j=1}^N \left( \frac{A(y_j)}{D} \right)^k  = \frac{1}{D^k} \left(\sum_{i=1}^M A(x_i)^k \pm \sum_{j=1}^N A(y_j)^k \right) \\
 \nonumber &= 
\frac{1}{D^k} \mathscr{S}^\pm(A^k(z)) .
\end{align}

The same computation for $R_k^{(x)} + R_k^{(y)}$ yields

\begin{equation*}
R^{(x)}_k + R^{(y)}_k = 
\frac{1}{D^k} \left(\sum_{i=1}^M x_iA(x_i)^k + \sum_{j=1}^N y_jA(y_j)^k \right) = 
\frac{1}{D^k} \mathscr{S}^+(zA^k(z)) .
\end{equation*}

We are now ready to state the main theorem of this section. 

\begin{theorem}\label{thm:InvSystem}
The solutions to the vortex problem satisfy, for every $k\geq 0$,
\begin{equation*}
\frac{1}{D} \mathscr{S}^+(zA^k(z)) = - \frac{1}{2} \sum_{i=0}^{k-1} \left( \mathscr{S}^-(A^i(z))\right)\left( \mathscr{S}^-(A^{k-i-1}(z))\right) + \frac{k}{2} \mathscr{S}^+(A^{k-1}(z)).
\end{equation*}
\proof
We rearrange the conjugate equation of Theorem~\ref{thm:invariant1}. 
\begin{equation}\label{Eq:ConjT21}
R_k^{(x)} +R_k^{(y)} =  
- \frac{1}{2} \sum_{i=0}^{k-1} (S_i^{(x)}S_{k-i-1}^{(x)} + S_i^{(y)}S_{k-i-1}^{(y)}) 
+ \sum_{i=0}^{k-1}S_{i}^{(x)}S_{k-i-1}^{(y)} 
+ \frac{k}{2}  (S_{k-1}^{(y)}+ S_{k-1}^{(x)}).
\end{equation}
We can also write 
\begin{multline*}
\sum_{i=0}^{k-1}S_{i}^{(x)}S_{k-i-1}^{(y)} =  \frac{1}{2}\left( \sum_{i=0}^{k-1}S_{i}^{(x)}S_{k-i-1}^{(y)} + \sum_{i=0}^{k-1}S_{i}^{(x)}S_{k-i-1}^{(y)} \right) =\\
 \frac{1}{2} \left( \sum_{i=0}^{k-1}S_{i}^{(x)}S_{k-i-1}^{(y)} + \sum_{i=0}^{k-1}S_{k-i-1}^{(x)}S_{i}^{(y)} \right) = \frac{1}{2} \sum_{i=0}^{k-1}(S_{i}^{(x)}S_{k-i-1}^{(y)} + S_{k-i-1}^{(x)}S_{i}^{(y)}).
\end{multline*}
This allows to rewrite the Equation~\eqref{Eq:ConjT21} as 
\begin{equation*}
R_k^{(x)} +R_k^{(y)} =  
- \frac{1}{2} \sum_{i=0}^{k-1} (S_i^{(x)}S_{k-i-1}^{(x)} + S_i^{(y)}S_{k-i-1}^{(y)} - S_{i}^{(x)}S_{k-i-1}^{(y)} - S_{k-i-1}^{(x)}S_{i}^{(y)}) 
+ \frac{k}{2}  (S_{k-1}^{(y)}+ S_{k-1}^{(x)}).
\end{equation*}
We have $R_k^{(x)} +R_k^{(y)}  = \frac{1}{D^k} \mathscr{S}^+(zA^k(z))$ and $S^{(x)}_{k-1} + S^{(y)}_{k-1} = \frac{1}{D^{k-1}} \mathscr{S}^+(A^{k-1}(z))$. 
Moreover we have 
\begin{multline*}
S_i^{(x)}S_{k-i-1}^{(x)} + S_i^{(y)}S_{k-i-1}^{(y)} - S_{i}^{(x)}S_{k-i-1}^{(y)} - S_{k-i-1}^{(x)}S_{i}^{(y)} = (S_i^{(x)} - S_i^{(y)}) (S_{k-i-1}^{(x)}  - S_{k-i-1}^{(y)}) \\ =\left(\frac{1}{D^i} \mathscr{S}^-(A^i(z))\right)\left(\frac{1}{D^{k-i-1}} \mathscr{S}^-(A^{k-i-1}(z))\right).
\end{multline*}
Equation~\eqref{Eq:ConjT21} now can be written as
\begin{align}\label{eqn:thm2.4}
 \frac{1}{D^k} \mathscr{S}^+&(zA^k(z)) \\ 
 \nonumber&=  
- \frac{1}{2 D^{k-1}} \sum_{i=0}^{k-1} (\left(\mathscr{S}^-(A^i(z))\right)\left( \mathscr{S}^-(A^{k-i-1}(z))\right)) 
+ \frac{k}{2D^{k-1}}  (  \mathscr{S}^+(A^{k-1}(z))),
\end{align}
and the statement follows at once.
\endproof
\end{theorem}

\begin{remark}\label{rmk:InvSystem}
For $k=0$ it reduces to $e_1^{(x)}+e_1^{(y)}=0$, for $k=1$ it reduces to $0=0$.
\end{remark}

%------------------------------------------------------------------------------------------------

\subsection{From solutions to the invariant system to solutions of (\ref{Eq:Zeroth})}

In this section we explain how to obtain solutions for the System (\ref{Eq:Zeroth}) using the symmetric system.

The first step will be to solve the invariant system obtained in Theorem \ref{thm:InvSystem}. In this case we will be aiming for complete, exact solutions. The main idea will be for each $k\geq 0$ to write
\begin{equation*}
h_k:=\frac{1}{D} \mathscr{S}^+(zA^k(z))  + \frac{1}{2} \sum_{i=0}^{k-1} \left( \mathscr{S}^-(A^i(z))\right)\left( \mathscr{S}^-(A^{k-i-1}(z))\right) - \frac{k}{2} \mathscr{S}^+(A^{k-1}(z)).
\end{equation*}
By construction the $h_k$'s are polynomial functions in the elementary symmetric polynomials $e_1^{(x)}, \ldots,e_M^{(x)}$ and $e_1^{(y)}, \ldots,e_N^{(y)}$. By Remark \ref{rmk:InvSystem}, we have $h_0=e_1^{(x)}+e_1^{(y)}$ and $h_1=0$. As the polynomial functions $e_1^{(x)}, \ldots,e_M^{(x)}$ and $e_1^{(y)}, \ldots,e_N^{(y)}$
are algebraically independent, we can think of them as coordinate functions on $\IC^{M+N}$. Our first step is then to find all points $p$ of $\IC^{M+N}$ which satisfy $h_k(p)=0$ for all $k$. As polynomial rings in finitely many variables are Noetherian, we know that there exists $K\geq 0$ such that the set of common zeros of $h_0,\ldots,h_K$ coincides with the set of common zeros of all $h_k$'s. Unfortunately, this finiteness result is not constructive, and it is very hard to determine a priori how many $h_k$'s are sufficient. Furthermore, the formulas for the $h_k$'s are rather complicated and each addition makes the symbolic computation less likely to be tractable. Accordingly we make a choice to consider only $h_0, h_2, \ldots, h_{M+N-1}$. This ensures that the number of equations is the same as the number of variables after eliminating $e^{(y)}_1$ via the relation $e^{(x)}_1+e^{(y)}_1=0$. The set of common zeros of these functions is potentially bigger than the set of common zeros of all the $h_k$'s, but we know this set will include the set of solutions to the System (\ref{Eq:Zeroth}).

\begin{proposition}
Equip $K[e_1^{(x)},\ldots,e_M^{(x)},e_1^{(y)},\ldots,e_N^{(y)}]$ with the structure of a graded ring by declaring $e_1^{(x)}$ and $e_j^{(y)}$ to have degree $i$ and $j$, respectfully. Then, for each $k\geq 0$ equation $h_k$ is homogeneous of degree $(k-1)\left({M \choose 2} + {N\choose 2}+2MN-1\right)$.
\end{proposition}
\begin{proof}
  As noted in the proof of Theorem \ref{thm:InvSystem}, equation
  (\ref{eqn:thm2.4}) is simply a rewriting of the conjugate of the equation in Theorem~\ref{thm:invariant1}. As this equation is homogeneous of degree $k-1$ in the variables $x_1,\ldots, x_M,y_1,\ldots,y_N$, the conjugate is homogenous of degree $1-k$ (indeed, Equations \ref{Eq:Exi} and \ref{Eq:Exi} express $X_i$ and $Y_i$ as homogeneous rational functions of degree $-1$ in the variables $x_1,\ldots, x_M,y_1,\ldots,y_N$).
  Next, we note that the discriminant 
\[D=\prod_{i\neq j}(x_i-x_j)(-1)^{MN}\prod_{i,j}(x_i-y_j)\prod_{i\neq j}(y_i-y_j)\]
is homogeneous of degree ${M \choose 2}+{N \choose 2}+2MN$ in the variables $x_1,\ldots, x_M,y_1,\ldots,y_N$. It follows that the equation $h_k$ is homogeneous of degree $(k-1)\left({M \choose 2} + {N\choose 2}+2MN-1\right)$ in the variables $x_1,\ldots, x_M,y_1,\ldots,y_N$. The conclusion then follows since $h_k$ is invariant and the grading in the statement of the proposition is the grading obtained by taking the degree of the elementary symmetric polynomials in the variables $x_1,\ldots, x_M,y_1,\ldots,y_N$.
\end{proof}

\begin{remark}\label{remark:S1toCstar}
The fact that the polynomials $h_k$ are homogeneous means that if $(e_1^{(x)},\allowbreak \ldots,\allowbreak e_M^{(x)}, e_1^{(y)},\ldots,e_N^{(y)})$ is a common zero, then so is $(te_1^{(x)},\ldots,\allowbreak t^Me_M^{(x)},\allowbreak  te_1^{(y)},\allowbreak \ldots,t^Ne_N^{(y)})$ for every $t\in \IC$. Equivalently, if $(x_1,\ldots,x_M,y_1,\ldots,y_N)$ is a solution, so is $(tx_1,\allowbreak \ldots,tx_M,\allowbreak ty_1,\allowbreak \ldots,\allowbreak ty_N)$ for every $t\in \IC$. This is not surprising. The action of the multiplicative group of complex numbers with modulus 1 on the solutions of Equations (\ref{Eq:Zeroth}) implies that for every solution $(x_1,\ldots,x_M,y_1,\ldots,y_N)$ the set $\{(tx_1,\allowbreak \ldots,\allowbreak tx_M,\allowbreak ty_1,\allowbreak \ldots,\allowbreak ty_N) \mid t\in \IC \text{ and } |t|=1\}$ is contained in the set of solutions of Equations (\ref{Eq:Zeroth}) and so in the common zeroes of the $h_k$'s. But as the $h_k$'s are polynomials, their common zeroes must contain the Zariski closure of this set, namely $\{(tx_1,\ldots,tx_M,ty_1,\ldots,ty_N) \mid t\in \IC \}$.
\end{remark}

\begin{question}
Are the sets of common zeros of System (\ref{Eq:Zeroth}) and the symmetric system from Theorem \ref{thm:InvSystem} finite  up to symmetry?
\end{question}

Zero-dimensional common zero sets of polynomials are finite, meaning that we could potentially list all solutions. With this in mind, we break down the problem into subclasses by dehomogenizing. Specifically, choosing an order on the variables, say starting with $e_1^{(x)}$, we divide the system into two cases, when $e_1^{(x)}=0$ and when $e_1^{(x)}\neq 0$. By Remark \ref{remark:S1toCstar} if there is a solution with $e_1^{(x)}\neq 0$, then up to multiplying by a complex number we can assume that $e_1^{(x)}=1$. Thus we add the equation $e_1^{(x)}=1$ to the subsystem.   Next we consider the case where $e_1^{(x)}$ is zero. Setting a variable to zero does not break the homogeneity and so we can again dehomogenize by setting the second variable to 1, say $e_2^{(x)}=1$. We continue like this until we run out of variables or the equations become trivial.  
  
  For every solution in the elementary symmetric polynomials we find one corresponding solution in the $x_i$'s and $y_j$'s. This is done by solving the system obtained by plugging the solution in the $e_i^{(x)}$'s and $e_i^{(y)}$'s in the equations defining  the elementary symmetric polynomials 
	$e_i^{(x)} = \sum_{j_1<\ldots<j_i} x_{j_1} \cdot \ldots \cdot x_{j_i}$ 
	and 
	$e_i^{(y)} = \sum_{j_1<\ldots<j_i} y_{j_1} \cdot \ldots \cdot y_{j_i}$.

          What we need to do next is on the one hand to remove the arbitrary choices we made when dehomogenizing (for example assuming $e_1^{(x)}=1$), and on the other hand check if the common zero of the $h_k$'s we obtained is a solution of the original System (\ref{Eq:Zeroth}). We do both at the same time. For every solution in the $x_i$'s and $y_j$'s, we are looking for $\lambda\in\mathbb{C}$ such that $\lambda x_1, \ldots, \lambda y_N$ satisfy system~(\ref{Eq:Zeroth}). Since we can scale by any complex number of modulus 1, we can restrict this search to $\lambda$ real and positive. Supposing $x_i\neq 0$, we use the $i$th equation of system~(\ref{Eq:Zeroth})
        rewritten as
\begin{eqnarray}\label{Eq:xirewrite}
  \begin{aligned}
	\lambda^2 = \frac{1}{\overline{x_i}}\left(- \sum_{\substack{j \neq i \\ j=1}}^M \frac{1}{x_i-x_j} + \sum _{j=1}^N \frac{1}{x_i-y_j}\right).
	  \end{aligned}
\end{eqnarray}
This allows us to determine the value which could work. We use the remaining equations of system~(\ref{Eq:Zeroth}) to check if this works for all. \\

\noindent\textbf{Summary of solution procedure:}\\
\noindent We start with the equations $h_0, h_2, \ldots, h_{M+N-1}, hD-1$.\\
\noindent For $i$ from $1$ to $M$:

\begin{enumerate}

	\item We dehomogenize the system by imposing the extra conditions 
	$e_1^{(x)} = \ldots =  e_{i-1}^{(x)}=0$, $e_i^{(x)} = 1$. \label{step1}
	
	\item We find the (finitely many) solutions to the subsystem obtained in Step (\ref{step1}), for example using Maple PolynomialSystem function.
         
	\item For each solution in the elementary symmetric polynomials, we find one corresponding solution in the $x_i$'s and $y_j$'s by solving the polynomial systems obtained by substituting each solution in the definition of the elementary polynomials as described above.
          
	\item For every solution in the $x_i$'s and $y_j$'s, we check wether $\lambda x_1, \ldots, \lambda y_N$ satisfy system~(\ref{Eq:Zeroth})  for some $\lambda>0$.

\end{enumerate}

%%%%%%%%%%%%%%%%%%%%%%%%%%%%%%%%%%%%%%%%%%%%%%%
%
%       SECTION: Discussion on computational advantages of this method
%
%%%%%%%%%%%%%%%%%%%%%%%%%%%%%%%%%%%%%%%%%%%%%%%

\section{Results and Benchmark}\label{Section:Results}

In this section, we present the solutions we have found and give details of the computations performed. 
The invariant equations described in the previous section are extremely long. We first attempted to write them in the software Macaulay2~\cite{M2}. This was,
however, extremely demanding in terms of processor time and memory requirements. The same operation was far more efficient using the software Maple~\cite{Maple}. One possible reason for this difference, is that Maple does not expand products of polynomial (in our case, the powers of $A$) unless required to do so. 

Standard numerical methods often require some initial guess and lead to one approximate solution. However, there can be multiple solutions to a given system of equations, and there is no guarantee that one will find all of them.
Given a system of polynomial equations, we use techniques developed in computational algebraic geometry and commutative algebra to compute \emph{all} the solutions. 
\begin{itemize}

	\item Symbolic AG: The most common symbolic method is based on the computation of a Gr\"obner basis for the system. Gr\"obner bases provide a systematic way to symbolically find the set of common zeroes of a system of polynomials.  Gr\"obner bases are (typically very long) lists of generators of the system of polynomial equations with good algebraic properties which can be understood as a multivariate generalization of Gaussian elimination. For more details see \cite{CoxLittleOShea}.

	\item Numerical AG: Numerical algebraic methods are based on a principle called ``homotopy continuation''. The system is put in a continuous deformation (a homotopy) to an appropriate ``known'' start system with similar properties. The solutions of the known system are tracked over $\mathbb{C}$ using homotopy continuation, which provides numerical approximations of all the distinct or isolated complex solutions of the original system, and these can be certified. By running homotopy continuation with appropriate generic homotopy parameter over $\mathbb{C}$, not $\mathbb{R}$, numerical algebraic geometry techniques with probability 1 find all solutions along the path \cite{coeffparam}. For details see \cite{SommeseWampler,Bertini}.

\end{itemize}
We attempted to compute the solutions of the vortex problem following both
of these approaches. As exact method, we used the function \texttt{PolynomialSystem} contained in the package \texttt{SolveTools} of the software Maple \cite{Maple}. 
We used Bertini \cite{Bertini}, an open numerical algebraic geometry software, which contains an implementation of homotopy continuation and numerically solves for all solutions. 

The table below compares the running time (in seconds) of the two approaches on the invariant system as well as on the systems described in Section~\ref{Section:direct}. The computation was performed on a standard office desktop (Intel i3 3.40GHz with 7.7GB ram).

\begin{table}
\begin{center}
    \begin{tabular}{ | p{4cm} || p{1.5cm} | p{1.5cm} | p{1.5cm} | p{3.5cm} |}
    \hline
    Method & $M=2$ & $M=2$  & $M=3$ & $M=3$  \\ 
     &  $N=1$ &  $N=2$ &  $N=1$ &$N=2$ \\\hline\hline
    Invariants equations with Maple & $1.6$ & $2.7$ & $2.2$ & $42.6$ \\ \hline
    Real and imaginary part with Maple & $0.2$ & $3.8$ & $1.7$ & ran out of memory \\ \hline
    Real and imaginary part with Bertini & $988$ & $>10000$ & $>10000$ & $>10000$ \\ \hline
    Conjugate variables with Maple & $0.4$ & $6.5$ & $2.4$ & ran out of memory \\ \hline
    Conjugate variables with Bertini & $2401$ & $>10000$ & $>10000$ & $>10000$ \\ \hline
    \end{tabular}
\end{center}
\caption{Running time in seconds of the two approaches on the invariant system as well as on the systems described in Subsection~\ref{Section:direct}. The computation was performed on a standard office desktop (Intel i3 3.40GHz with 7.7GB ram).}\label{table:benchmark}
\end{table}

\subsection{List of solutions}

The table below contains all the solutions of the system for different values of $M$ and $N$. For $M=3, N=3$ and for $M=4, N=2$ the computation of the main component ($e_1^{(x)}=1$) did not terminate. However we computed all solutions where at least one variable $e^{(x)}_i$ or $e^{(y)}_i$ equals zero. For all other values of $M$ and $N$ in the table below, the computation was completed, providing computational proof that no more solutions are present. 

\begin{table}
\begin{center}
    \begin{tabular}{ | p{1.10cm} || p{1.35cm} | p{1.35cm}| p{1.35cm} | p{1.35cm} | p{1.35cm}| p{1.35cm}| p{1.35cm}| p{1.35cm} |}
    \hline
     & $M=1$ & $M=2$ & $M=2$ & $M=3$ & $M=3$ & $M=4$ & $M=3$ & $M=4$  \\ 
     & $N=1$ & $N=1$ & $N=2$ & $N=1$ & $N=2$ & $N=1$ & $N=3$ & $N=2$ \\
 \hline\hline $|\text{sol's}|$ 
     & $1$   & $1$   & $2$   & $0$   & $1$   & $0$   & $\geq 2$ & $\geq 1$\\ \hline
    \end{tabular}
\end{center}
\caption{Number of solutions found for different values of $M$ and $N$. For $M=3, N=3$ and for $M=4, N=2$ there may be more solutions.}
\end{table}

The table below contains all the solutions found.

\begin{table}
\begin{center}
	\begin{tabular}{ | p{1.5cm} || p{5cm}  p{5cm} |}
	\hline
	$(1,1)$ &	$x_1 = -\sqrt{2}/2 \approx -0.707$ & $y_1 = \sqrt{2}/2 \approx -0.707$ \\ 
	\hline
	$(2,1)$ &	$x_1 = -\sqrt{2}/2 \approx -0.707$ & $y_1 = 0$ \\
			 &	$x_2 = \sqrt{2}/2 \approx -0.707$ &  \\ 
	\hline
	$(2,2)$ &	$x_1 = -\sqrt{2}/2 \approx -0.707$ & $y_1 = -\sqrt{2}/2 i \approx -0.707 i $ \\
		    &	$x_2 = \sqrt{2}/2 \approx 0.707$ & $y_2 = \sqrt{2}/2 i \approx 0.707 i$  \\ 
	\hline
	$(2,2)$ &	$x_1 = \frac{\sqrt{\sqrt{2}+\sqrt{1+\sqrt{2}}}}{2}\cdot\frac{\sqrt{2}}{(1+\sqrt{2})^{\frac{1}{4}}} \approx 0.977$ & $y_1 \approx -0.977 $ \\
		    &	$x_2 \approx -0.212$ & $y_2  \approx 0.212 $  \\ 
	\hline
	$(3,2)$ &	$x_1 \approx 0$ & $y_1 = \approx -0.366 $ \\
		    &	$x_2 \approx -0.930$ & $y_2  \approx 0.366 $  \\ 
		    &	$x_3 \approx 0.930$ &  \\
	\hline
	$(3,3)$ &	$x_1 \approx 0.707$ & $y_1 = \approx -0.707 $ \\
		    &	$x_2 \approx -0.354+0.612 i$ & $y_2  \approx 0.353+0.612 i$ \\ 
		    &	$x_3 \approx -0.354-0.612 i$ & $y_2  \approx 0.353-0.612 i$ \\
	\hline
	$(3,3)$ &	$x_1 \approx -0.476$ & $y_1 = \approx 0.476 $ \\
		    &	$x_2 \approx 0.162$ & $y_2  \approx -0.162 $  \\ 
		    &	$x_3 \approx 1.112$ & $x_3 \approx -1.112$ \\
    \hline
	$(4,2)$ &	$x_1 \approx -0.600 i$ & $y_1 = \approx -0.285 i $ \\
		    &	$x_2 \approx -0.241$ & $y_2  \approx 0.285 i $  \\ 
		    &	$x_3 \approx 0.241$ &  \\
		    &	$x_3 \approx 0.600 i$ &  \\
	\hline
    \end{tabular}
\end{center}
\caption{List of solutions found. The positions of the positively
  charged $x_i$ and negatively charged $y_i$ vortices are provided
  that solve Equations~(\ref{Eq:Zeroth}).}
\end{table}
\begin{figure}[p]
\begin{tabular}{c:c}
\begin{minipage}{0.3\textwidth}
(a)\qquad
	\begin{tikzpicture}[framed, x=.85cm,y=.85cm]
		\draw[-stealth, thin, draw=gray] (-1.5,0)--(1.5,0); 
		\draw[-stealth, thin, draw=gray] (0,-1.5)--(0,1.5); 
		\node [red] at (1,0) {\textbullet};
		\node [red] at (-1,0) {\textbullet};
		\node [blue] at (0,0) {\textbullet};
	\end{tikzpicture}
\end{minipage}
&
\begin{minipage}{0.55\textwidth}
(b)\qquad
	\begin{tikzpicture}[framed, x=.85cm,y=.85cm]
		\draw[-stealth, thin, draw=gray] (-1.5,0)--(1.5,0); 
		\draw[-stealth, thin, draw=gray] (0,-1.5)--(0,1.5); 
		\node [red] at (1,0) {\textbullet};
		\node [red] at (-1,0) {\textbullet};
		\node [blue] at (0,1) {\textbullet};
		\node [blue] at (0,-1) {\textbullet};
	\end{tikzpicture}
	\hspace{5pt}
	\begin{tikzpicture}[framed, x=.85cm,y=.85cm]
		\draw[-stealth, thin, draw=gray] (-1.5,0)--(1.5,0); 
		\draw[-stealth, thin, draw=gray] (0,-1.5)--(0,1.5); 
		\node [red] at (1,0) {\textbullet};
		\node [red] at (-.217,0) {\textbullet};
		\node [blue] at (-1,0) {\textbullet};
		\node [blue] at (.217,0) {\textbullet};
	\end{tikzpicture}
\end{minipage}
\\
&\\
\hdashline
&\\
\begin{minipage}{0.3\textwidth}
(c)\qquad
	\begin{tikzpicture}[framed, x=.85cm,y=.85cm]
		\draw[-stealth, thin, draw=gray] (-1.5,0)--(1.5,0); 
		\draw[-stealth, thin, draw=gray] (0,-1.5)--(0,1.5); 
		\node [red] at (1,0) {\textbullet};
		\node [red] at (-1,0) {\textbullet};
		\node [red] at (0,0) {\textbullet};
		\node [blue] at (.387,0) {\textbullet};
		\node [blue] at (-.387,0) {\textbullet};
	\end{tikzpicture}
\end{minipage}
&
\begin{minipage}{0.55\textwidth}
(d)\qquad
	\begin{tikzpicture}[framed, x=.85cm,y=.85cm]
		\draw[-stealth, thin, draw=gray] (-1.5,0)--(1.5,0); 
		\draw[-stealth, thin, draw=gray] (0,-1.5)--(0,1.5); 
		\node [red] at (1,0) {\textbullet};
		\node [red] at (-.5,-.866) {\textbullet};
		\node [red] at (-.5,.866) {\textbullet};
		\node [blue] at (-1,0) {\textbullet};
		\node [blue] at (.5,.866) {\textbullet};
		\node [blue] at (.5,-.866) {\textbullet};
	\end{tikzpicture}
	\hspace{5pt}
	\begin{tikzpicture}[framed, x=.85cm,y=.85cm]
		\draw[-stealth, thin, draw=gray] (-1.5,0)--(1.5,0); 
		\draw[-stealth, thin, draw=gray] (0,-1.5)--(0,1.5); 
		\node [red] at (1,0) {\textbullet};
		\node [red] at (-.427,0) {\textbullet};
		\node [red] at (.145,0) {\textbullet};
		\node [blue] at (-1,0) {\textbullet};
		\node [blue] at (.427,0) {\textbullet};
		\node [blue] at (-.145,-0) {\textbullet};
	\end{tikzpicture}
\end{minipage}
\\
&\\
\hdashline
&\\
\begin{minipage}{0.3\textwidth}
(e)\qquad
	\begin{tikzpicture}[framed, x=.85cm,y=.85cm]
		\draw[-stealth, thin, draw=gray] (-1.5,0)--(1.5,0); 
		\draw[-stealth, thin, draw=gray] (0,-1.5)--(0,1.5); 
		\node [red] at (1,0) {\textbullet};
		\node [red] at (-1,0) {\textbullet};
		\node [red] at (0,-.4) {\textbullet};
		\node [red] at (0,.4) {\textbullet};
		\node [blue] at (-.467,0) {\textbullet};
		\node [blue] at (.457,0) {\textbullet};
	\end{tikzpicture}
\end{minipage}
&
\end{tabular}
\caption[]{The solutions to the vortex problem for: $M=2$, $N=1$, (a); $M=2$, $N=2$, (b); $M=3$, $N=2$, (c); $M=3$, $N=3$, (d); $M=4$, $N=2$, (e). }
\end{figure}
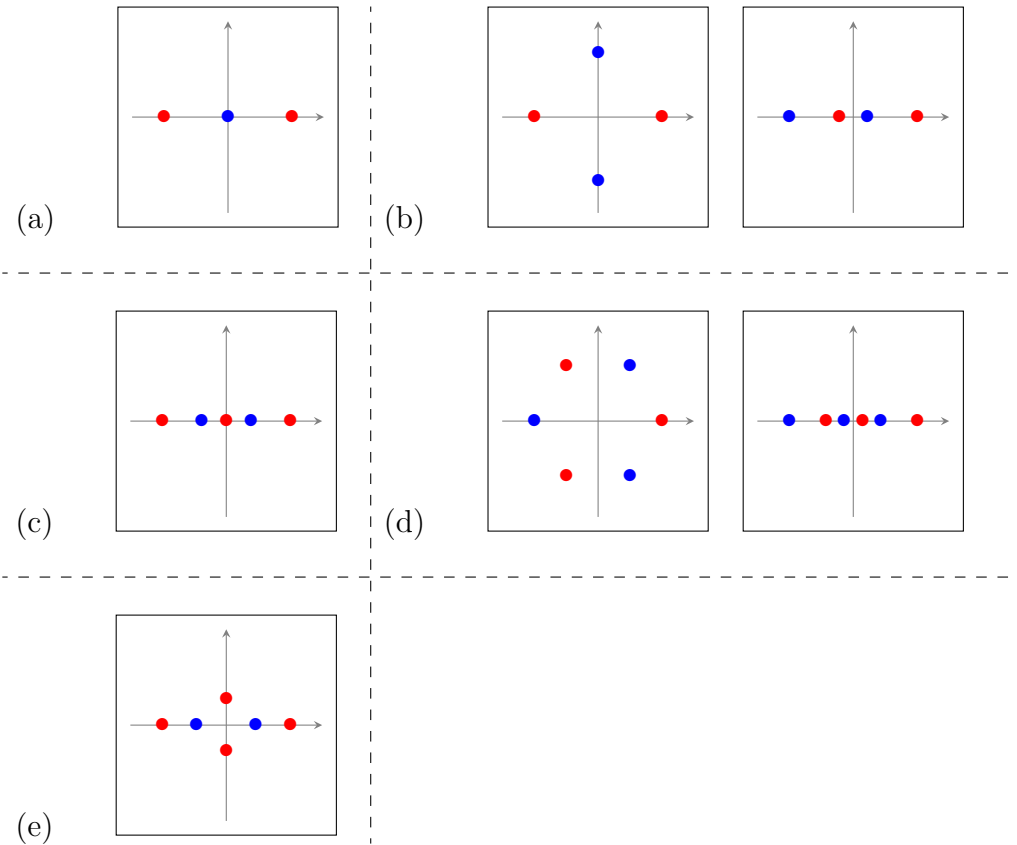

\section{Connection to the GP PDE Results}\label{Section:BacktoPDE}

Many of the above obtained configurations have been
previously identified at the level of the GP equation. In particular,
for instance, the vortex dipoles ($M=N=1$)
have emerged as the lowest
order configuration that destabilizes a planar dark soliton
state~\cite{komi,midde,cluster} and have also been obtained
experimentally via different techniques~\cite{NeelyEtAl10,dshall1},
enabling the observation of their precessional dynamics.
Importantly, in~\cite{dshall1}, the stationary form of the configuration
directly related to the considerations herein, was also experimentally
identified.
Furthermore, in some of these works~\cite{midde,cluster}, it
was argued that the aligned configurations of the tripole
with $M=2$, $N=1$ (which was also observed
experimentally in~\cite{bagno}), the aligned quadrupole with $M=2$, $N=2$,
then the aligned states with $M=3$, $N=2$, as well as that with $M=3$, $N=3$
(and so on) are all byproducts of subsequent progressive
further destabilizations of the dark soliton stripe. That is,
for such a stripe~\cite{midde}, each additional destabilization
produces a stationary configuration with one additional vortex
along the former dark line soliton. This is an intriguing cascade
of bifurcations from the stripe which explains the emergence
of aligned alternating charge vortex configurations, each with
one additional charge with respect to the previous one. Each
of these arises through a (supercritical) pitchfork bifurcation
which, in turn, justifies that each of these has an additional unstable
eigendirection with respect to the previous one.
Consequently, the vortex dipole is the most robust among these
configurations, bearing no real eigenvalues (and no exponential
instabilities), but only an internal, potentially resonant via a Hamiltonian
Hopf bifurcation, mode in the system. Then the tripole would bear
one exponentially unstable eigendirection, the aligned quadrupole
two such, and so on.

It is important here to highlight that some of
the early existence and even stability results on the subject
were obtained in the works of~\cite{dumi,motto1,motto2}.
In these works, in addition to some of the aligned configurations,
including the dipole and tripole, the first example of a canonical
polygon of alternating vortices, namely the quadrupole was
identified. It was, in fact, found that this configuration too
did not bear any exponential instabilities but could become
unstable through an oscillatory instability. The work of~\cite{cluster}
offered a more systematic viewpoint on these polygonal configurations
(see also~\cite{barry1}). There, it was found that these states
too were a result of the destabilization of a dark solitonic stripe,
but this time a radial one, the so-called ring dark soliton or RDS
configuration
(first proposed in the BEC context in~\cite{RDS}). In particular,
as soon as this state emerges (in the linear limit of the system)
it is degenerate with the vortex quadrupole. Then, its next
(further) destabilizing bifurcation gives birth to a vortex
hexagon, the subsequent one to a vortex octagon, then to a decagon
and so on. All of these lead to canonical polygons involving
alternating pairs of vortices, each of which has one more (again)
unstable eigendirection to the previous one, i.e., the hexagon
is generically unstable due to pairs of eigenvalues
emerging as a result of the destabilization
of the RDS. Moreover, the method
of generating functions was used to illustrate that
states with $M=N$, can be used to construct polygons of
angle $\phi=\pi/N$ (at a fixed radius) between the
alternating charges.

To give a canonical example in the context of configurations
considered herein, we briefly refer to the case of the hexagon.
In Figure~\ref{fign1}, we provide a typical scenario involving
the case of $\mu=2$ and $\Omega=0.1$. We consider the
different layers of approximation, starting with Equation~(\ref{aux1}),
which is the one also tackled via our algebraic techniques.
At that level, as is established in~\cite{barry2} (and also
found herein) the positions of the vortices are cube roots
of unity for both the positive and negative charges,
displaced by $\pi/3$ with respect to each other. The radius
of the solutions, as shown also in Table 3 (for the complex
(3,3) roots) is $\sqrt{2}/2 \approx 0.707$. This radius
is given in units of the Thomas-Fermi radius $R_{TF}=\sqrt{2 \mu}/\Omega$.
As seen in Figure~\ref{fign1}, the realistic radius is closer to $0.35 R_{TF}$
in the full numerical (PDE) computations. This difference is reflected
by the more accurate nature of Equations~(\ref{aux3}) and~(\ref{aux4}).
It is worthwhile to note that given the equidistant
from the origin nature of this configuration as regards
the vortices, in this case these two equations
[(\ref{aux3}) and~(\ref{aux4})] yield the same prediction.
For both of them, the equilibrium radius is found to
be $R^2=(2 \omega_{pr}(0)+R_{TF}^{-2})^{-1}$. For the parameters
above, $R=7.564=0.378 R_{TF}$; notice that this is very close
to the numerical result, the difference being justified by the
deviation of the above $\mu=2$ scenario from the Thomas-Fermi
limit of large values of $\mu$. Nevertheless, it is clear that the
qualitative picture is accurate in all the effective particle
descriptions and that the improved models can yield an even
quantitatively accurate characterization. It is relevant to note
that the stability results of Figure~\ref{fign1} indicate that this
is an unstable configuration due to two nearly identical pairs of
real eigenvalues, suggesting an exponential growth of perturbations
along the corresponding eigendirections.

All of the above configurations have also been summarized in the
compendium of~\cite{siambook} and it is interesting to note that
they include {\it all} the configurations that we have
obtained in the present work {\it except for} the $M=4$,
$N=2$ state of Figure~5. 
It is thus the latter that we now turn our
attention to more systematically, as it is unprecedented in earlier
both existence and stability
studies, to the best of our knowledge. This configuration
consists of 4 plus and 2 minus (or vice versa) charged vortices
with the inner ones constituting a quadrupole
--with slightly unequal distances from the origin along the
two axes--,
while the last
two are aligned with one of the axes and oppositely
charged to the rest of the vortices along the same line.
A typical example of this configuration was obtained and is
shown in Figure~\ref{fign2}. Importantly, in this example
our numerical observations once again bear some difference
in comparison to the
prediction of Table 3. In particular, numerically we find
the vortices to be located as follows. The $+1$ charges are
located at $x_{1,2}=\pm 0.1875 R_{TF} i$ (along the
imaginary axis) and at $x_{3,4}=0.3875 R_{TF}$,
while the $-1$ charges are at $y_{1,2}=\pm 0.15 R_{TF}$ (cf. with the
values in Table 3). The latter four vortices are located
on the real axis, although clearly the entire configuration
is freely amenable to azimuthal rotations. 
It is for this reason that we now
resort to the progressively more accurate representations
of Equations~(\ref{aux3}) [which accounts for the radial dependence
  of the vortex precession frequency] and then Equations~(\ref{aux4})
[which additionally incorporates  the effect of the inhomogeneity of the
  background in the inter-vortex interaction]. The former yields a prediction
of $x_{1,2}=\pm 0.158 R_{TF} i$ and $x_{3,4}=\pm 0.331 R_{TF}$
while $y_{1,2}=\pm 0.134 R_{TF}$. Finally, the most
accurate available description of Equations~(\ref{aux4})
leads to the following numbers $x_{1,2}=\pm 0.179 R_{TF} i$,
$x_{3,4}=0.413 R_{TF}$, while $y_{1,2}=\pm 0.154 R_{TF}$. This latter
description is the most accurate one --to the best of
our knowledge-- that is obtained by a particle model, being
limited only by the deviation from the Thomas-Fermi limit. That is,
the prediction would only be better for larger values of $\mu$.
Nevertheless, the conclusion that we reach
is that the configuration predicted
by the computer algebra techniques is qualitatively consonant with
the configuration identified in the full system. Nevertheless,
the more elaborate (and more accurate) models such as ultimately
that of Equations~(\ref{aux4}) are needed in order to most adequately
capture the quantitative specifics of the vortex locations.
In that light, the tools developed herein can be useful in
unraveling configurations possibly with quite limited symmetry
characteristics which may not be easily identifiable differently.
It should also be added that we have explored the dynamical
stability of this configuration and have found it to be
dynamically unstable, as shown in Figure~\ref{fign2}. In particular,
as can be observed in the figure, it bears two pairs of
real eigenvalues and a complex eigenvalue quartet with nontrivial
growth rates. The former lead to an exponential growth along the
respective eigendirections, while the latter corresponds
to an oscillatory growth due to the complex nature of the eigenvalues.

\begin{figure}[ht]
\centerline{
\includegraphics[height=4.1cm]{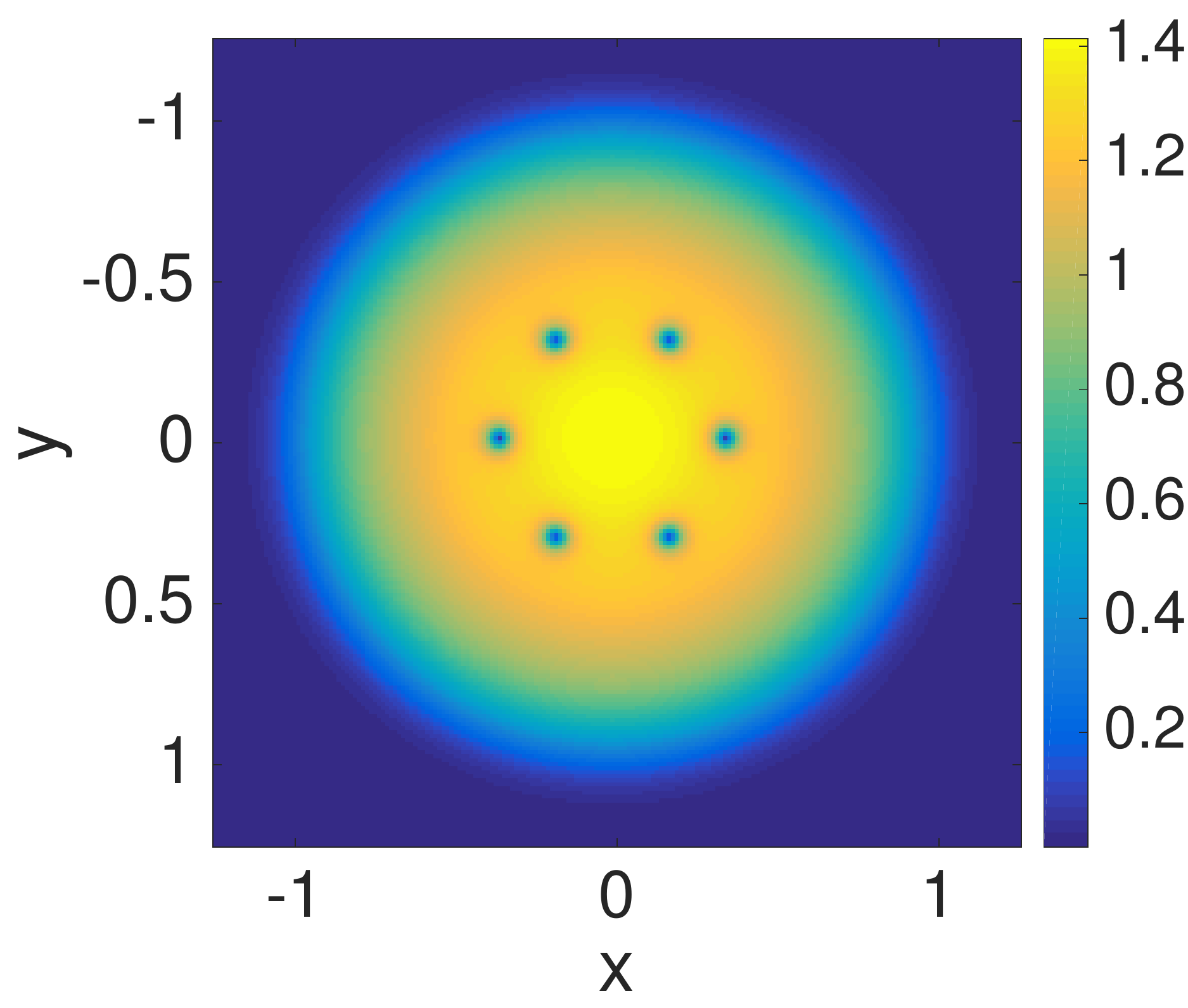}
\includegraphics[height=4.1cm]{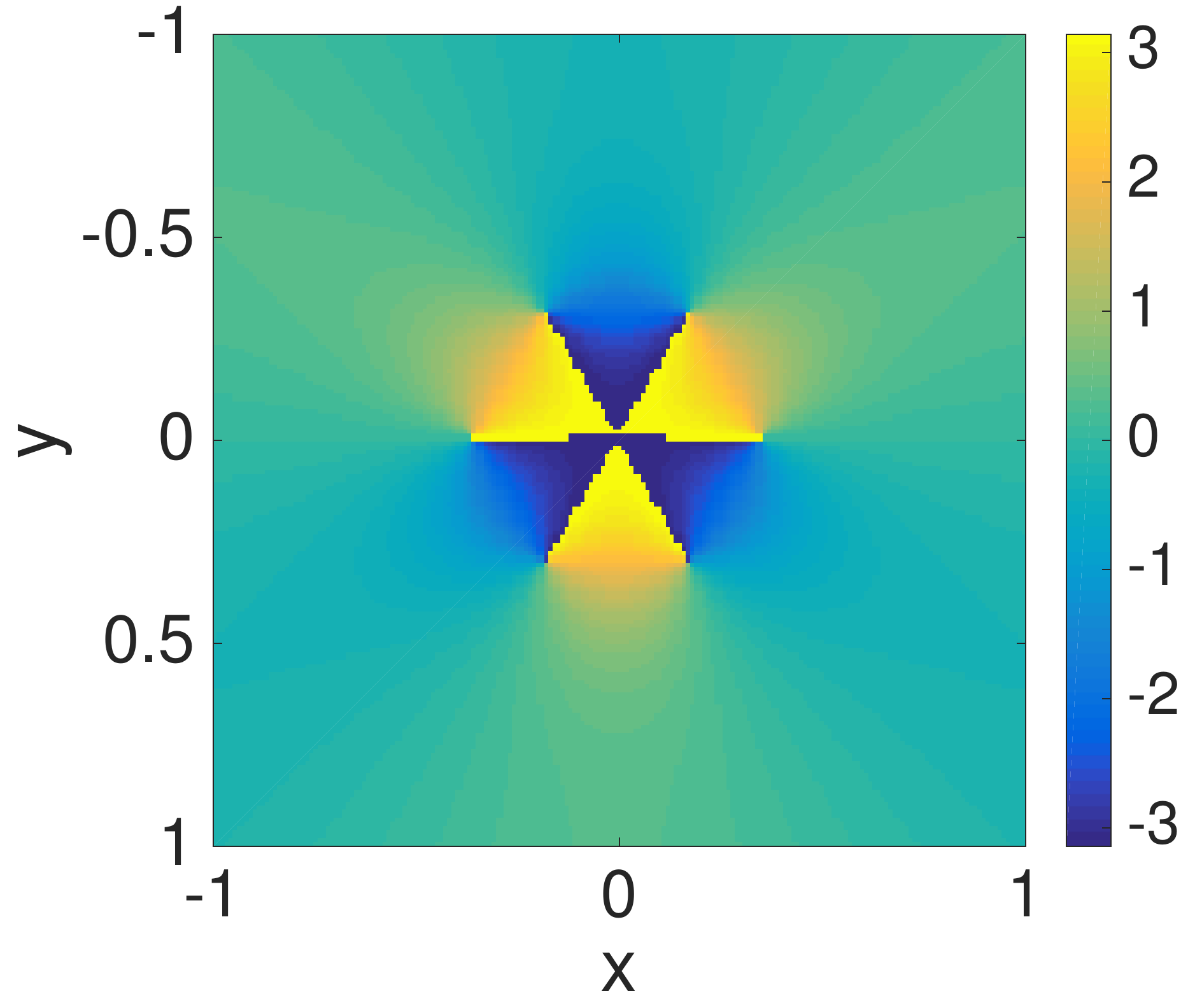}
\includegraphics[height=4.1cm]{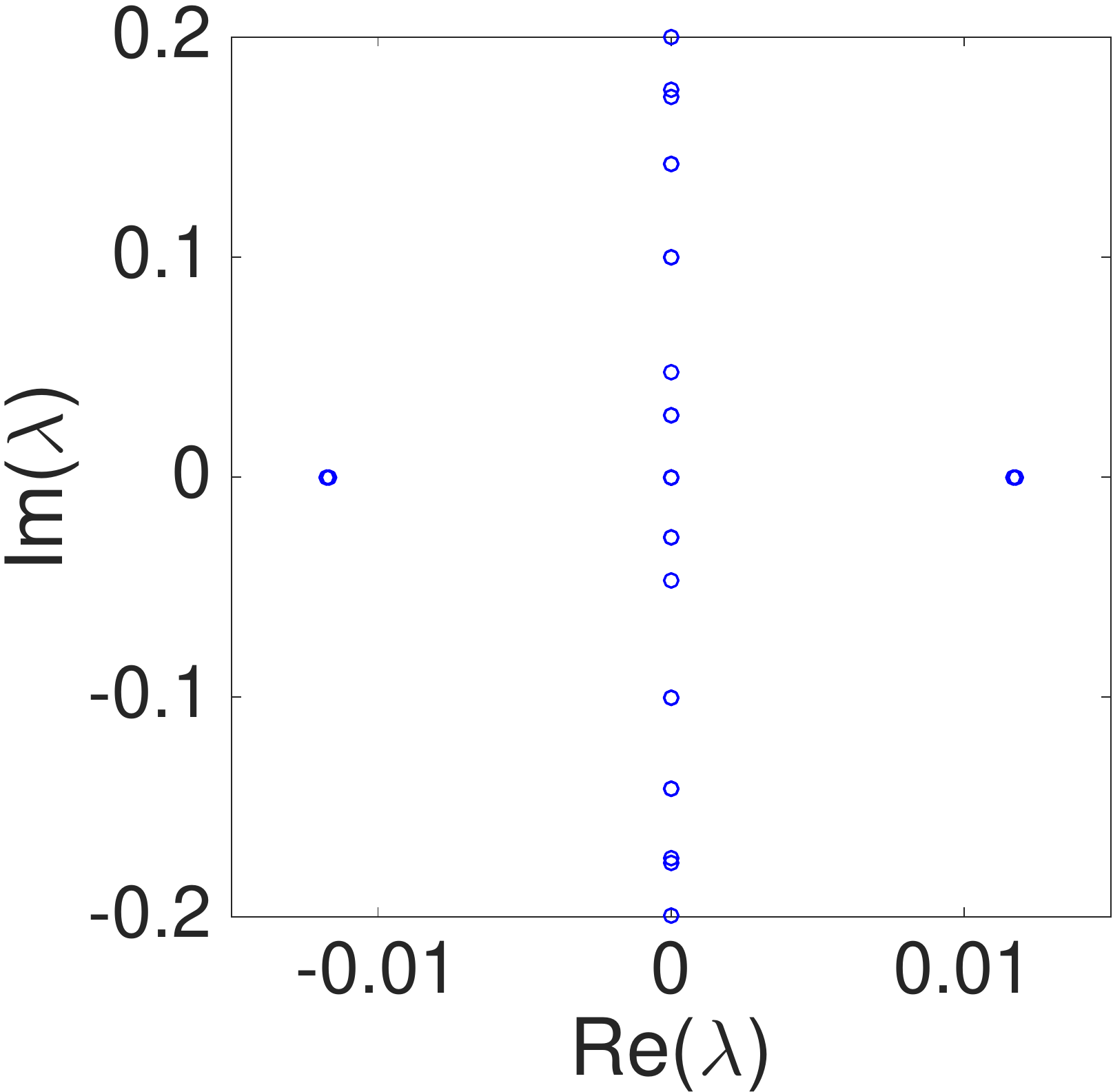}
}
\caption[]{The first panel shows the density field $|u|^2$ from
  a PDE computation of Equation~(\ref{GPE}) involving $M=N=3$ vortices
  in a hexagonal configuration. This is shown by a two-dimensional
  contour plot in the $(x,y)$ plane.
  The second panel illustrates
  the corresponding phase, revealing the alternating nature
  of the charges. Finally, the third panel illustrates the
  results of the linear stability analysis 
  around such a configuration. The spectral plane $(\lambda_r,\lambda_i)$
  of the associated
  linearization is illustrated for the corresponding eigenvalues
  $\lambda=\lambda_r + i \lambda_i$. The presence of two pairs of
  purely real eigenvalues establishes the exponential instability
  of such a solution.
}
\label{fign1}
\end{figure}

 \begin{figure}[ht]
\centerline{
\includegraphics[height=4.1cm]{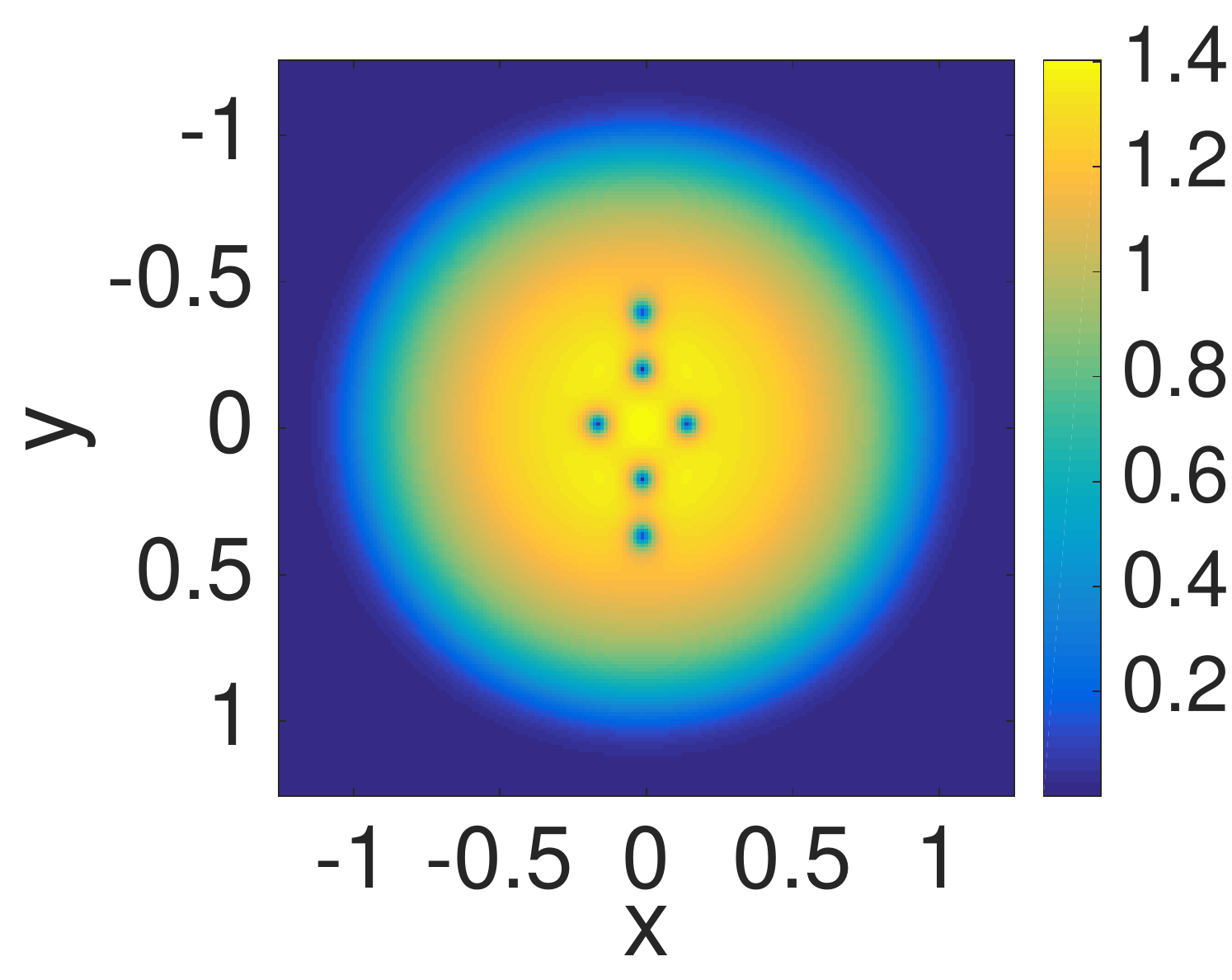}
\includegraphics[height=4.1cm]{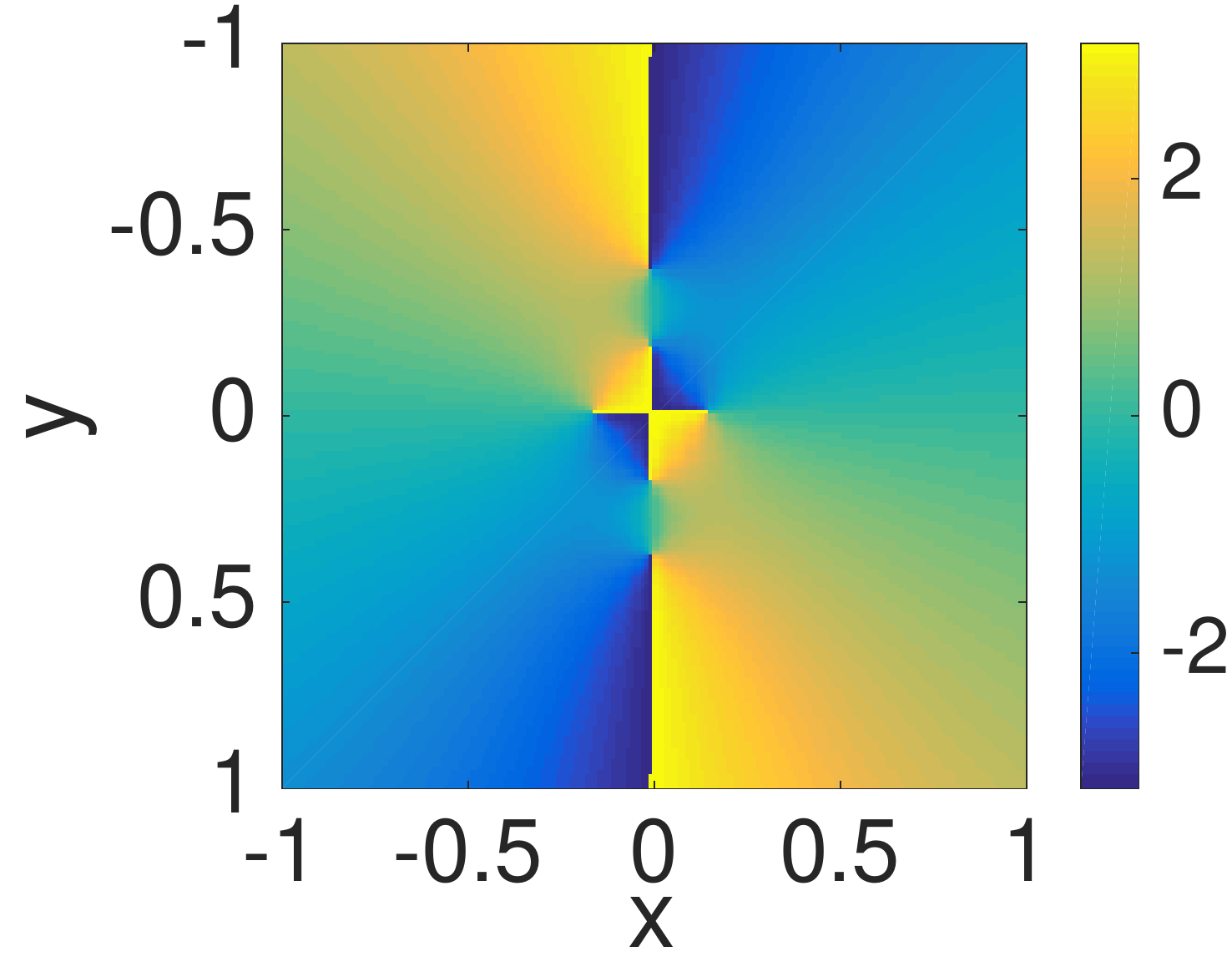}
\includegraphics[height=4.1cm]{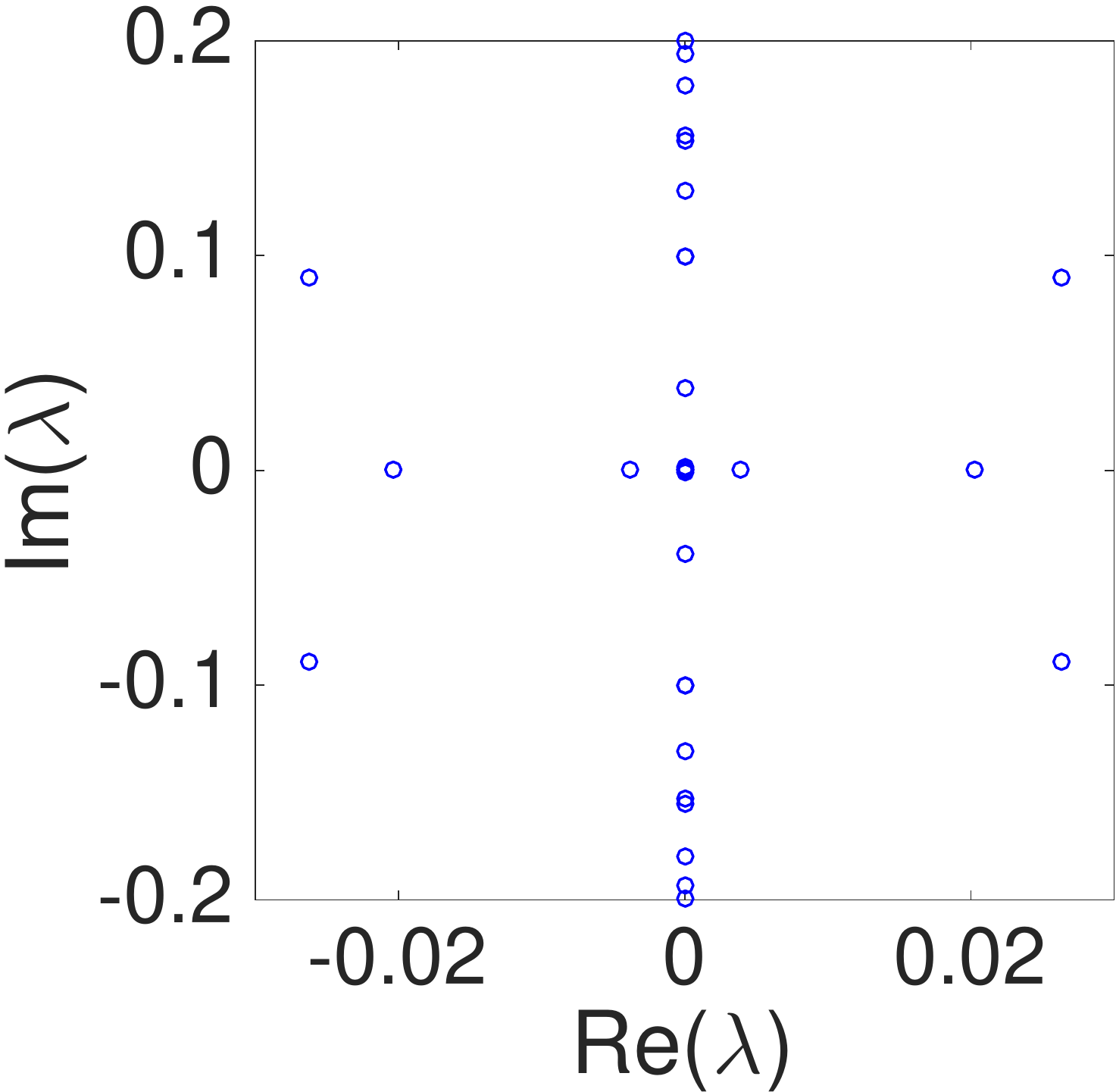}
}
\caption[]{Same as the previous configuration, but now for the
  principal configuration discovered herein, namely the
  stationary solutions involving $M=4$ positive and $N=2$
  negative charges. For details regarding the positions of
  the vortices and the comparison with the corresponding theoretical
  prediction, see the text. The last panel showcases the
  instability of this newly established configuration, by virtue
  of showing its two pairs of real eigenvalues (exponential
  instabilities) and one pair of complex eigenvalues
  (oscillatory instability).}
\label{fign2}
\end{figure}

\section{Conclusions \& Future Challenges}

In the present work, we have made an attempt to bring to bear
tools from the theory of Gr{\"o}bner bases and associated
computational algebra to the case of a problem involving
stationary configurations of oppositely charged 
vortices in atomic Bose-Einstein condensates. More specifically,
we have started from the corresponding PDE system (of the Gross-Pitaevskii
type) and have discussed different layers of reduction approximations
characterizing the dynamics of the vortices. The first layer is a
quasi-homogeneous one, the next involves the dependence of the precession
frequency of a single vortex on the distance, while the most elaborate
one also accounts for the inhomogeneity of the background in affecting
inter-vortex interactions. For computational simplicity reasons, we
have utilized the simplest one of these descriptions and deduced from
its corresponding steady state problem, a set of equations in the
elementary symmetric polynomials in the variables. We have brought
to bear computational algebra packages that have enabled us, in this
adapted formulation, to find {\it all} possible stationary
configurations involving
up to $5$ vortices of combined positive and negative charges and some
prototypical stationary configurations involving $6$ vortices.
Most configurations among these have been already obtained in the
literature of BECs, most notably the configurations with high
symmetry (axial ones with all the vortices on a line, or polygonal
ones with them sitting at the vertices of a canonical polygon).
We have discussed in some further detail one of these cases, namely
a canonical hexagon, consisting of 3 plus and 3 minus charged vortices.
{\it Nevertheless}, already in the case of $6$ charges, we have
presented an unprecedented --to the best of our knowledge--
configuration, namely one with 4 positive
and 2 negative charges (or vice versa). 
We have
studied such a configuration at the level of our different layers
of ODE approximation in comparison with computations of the original
PDE. In all the cases considered, while admittedly we have utilized
(for computational simplicity) the computer algebra package in the simplest
setting of Equation~(\ref{aux1}), we have found that the identified
configurations in {\it all} cases, persist in the full PDE problem.
Additionally, we have shown how the more adequate (but at the same
time more complex) polynomial equations of the models of Equations~(\ref{aux3})
and~(\ref{aux4}) can then facilitate a more accurate capturing
of the precise vortex locations in connection with the full PDE problem,
providing in this way a more definitive characterization of the states.

We believe that this larger scale program has numerous directions
for potential further development. A natural question concerns
the limitations of this effort regarding the attempt to seek
all the possible configurations with higher
numbers of charges. Presently, this task seems somewhat
limited by computational capabilities, but it seems reasonably
likely that advances in either the algorithmic developments or the
computational hardware may be able in the near future to
circumvent this issue and offer us an unprecedented ability
to obtain {\it all} stationary vortex configurations of
higher numbers of charges. A complementary effort may be developed
in the direction of bringing to bear similar algebraic techniques
but for the more complex and thus more cumbersome systems, such as
those of Equations~(\ref{aux3}) and especially so Equations~(\ref{aux4}).
Finally, there are numerous additional directions where one can
extend present considerations. While a large vein of potential
work can be opened by considering three-dimensional settings, we
limit our considerations to the 2d case, but involving potentially
traveling configurations. There exist works such as those
of~\cite{kadtke} and more recently~\cite{wei} which have discussed
intriguing algebraic connections including those with the so-called
Adler-Moser polynomials (see also references therein). Nevertheless
one can envision important, physically relevant variations where
the vortices are confined in one direction in the plane, while traveling
in the other direction. There, it is conceivable that the nice
algebraic structure of the Adler-Moser polynomials disappears, yet
a computer algebra characterization of steadily propagating solutions
may well be possible. These different directions are
currently under consideration and progress along them will be
reported in future publications.

%%%%%%%%%%%%%%%%%%%%%%%%%%%%%%%%%%%%%%%%%
%%%%%%%%%%%%%%%%%%%%%%%%%%%%%%%%%%%%%%%%%

\vspace{5mm}

{\it Acknowledgements.} Initial work on this project was funded by the John Fell Oxford University Press (OUP) Research Fund. H.A.H. acknowledges support from a Royal Society University Research Fellowship. H.A.H. thanks J. Hauenstein for useful discussions about numerical algebraic geometry. E.D.  and P.T. were supported by the University of Nottingham via a Anne McLaren Fellowship.
P.G.K. acknowledges useful discussions on
the subject with R. Carretero. This material is based upon work supported 
by the National Science Foundation under Grant No. DMS-1602994 (P.G.K.).

\bibliographystyle{plain}
\bibliography{vortices_refs}

\end{document}